\documentclass[aps,pra,twocolumn,superscriptaddress,longbibliography]{revtex4-2}

\usepackage{amssymb,amsmath,amsthm,bbm,bbold}
\usepackage{amsfonts}
\usepackage{graphicx}
\usepackage{mathtools}
\usepackage{hyperref}
\usepackage{float}

\usepackage[normalem]{ulem}
\usepackage{color}
\usepackage{xcolor}

\def\RR{\mathbbm{R}}

\def\H{\mathcal{H}}
\def\Ho{\mathcal{H}_1}
\def\HN{\wedge^N\Ho}

\newcommand{\F}{\mathcal{F}}
\newcommand{\bd}[1]{\boldsymbol{#1}}
\newcommand{\wb}{\bd{w}}
\newcommand{\g}{\gamma}

\newcommand{\G}{\Gamma}

\newcommand*\xbar[1]{\hbox{\vbox{
       \hrule height 0.6pt 
       \kern0.3ex
       \hbox{%
         \kern-0.2em
         \ensuremath{#1}%
         \kern 0.0em
         }}}}

\newcommand*\xxbar[1]{\hbox{\vbox{
       \hrule height 0.6pt 
       \kern0.3ex
       \hbox{%
         \kern-0.0em
         \ensuremath{#1}%
         \kern 0.0em
         }}}}

\newcommand{\Fw}{\mathcal{F}_{\!\wb}}
\newcommand{\Fbw}{\,\xbar{\mathcal{F}}_{\!\wb}}
\newcommand{\Gbw}{\,\xbar{\mathcal{G}}_{\!\wb}}

\newcommand{\E}{\mathcal{E}^N}
\newcommand{\p}{\mathcal{P}^1_N}
\newcommand{\e}{\mathcal{E}^1_N}
\newcommand{\Ew}{\mathcal{E}^N\!(\wb)}
\newcommand{\ew}{\mathcal{E}^1_N(\wb)}

\newcommand{\Ebw}{\,\xxbar{\mathcal{E}}^N\!(\wb)}
\newcommand{\ebw}{\,\xxbar{\mathcal{E}}^1_N(\wb)}

\newcommand{\ebwS}{\,\scriptsize{\xbar{\mathcal{E}}}\normalsize^1_N\hspace{-0.3mm}(\wb)}
\newcommand{\EbwS}{\,\scriptsize{\xbar{\mathcal{E}}}\normalsize^N\hspace{-0.7mm}(\wb)}

\newcommand{\Tr}{\mathrm{Tr}}
\DeclareMathOperator*{\argmin}{arg\,min}

\newcommand{\bra}[1]{\mbox{$\langle #1 |$}}
\newcommand{\ket}[1]{\mbox{$| #1 \rangle$}}

\newtheorem{thm}{Theorem}

\begin{document}
\title{Foundation of one-particle reduced density matrix functional theory for excited states}

\author{Julia Liebert}
\affiliation{Department of Physics, Arnold Sommerfeld Center for Theoretical Physics,
Ludwig-Maximilians-Universität München, Theresienstrasse 37, 80333 München, Germany}
\affiliation{Munich Center for Quantum Science and Technology (MCQST), Schellingstrasse 4, 80799 M\"unchen, Germany}

\author{Federico Castillo}
\affiliation{Max Planck Institute for Mathematics in the Sciences, Inselstraße 22, 04103 Leipzig, Germany}

\author{Jean-Philippe Labb\'e}
\affiliation{Department of Physics, Arnold Sommerfeld Center for Theoretical Physics,
Ludwig-Maximilians-Universität München, Theresienstrasse 37, 80333 München, Germany}
\affiliation{Munich Center for Quantum Science and Technology (MCQST), Schellingstrasse 4, 80799 M\"unchen, Germany}
\affiliation{Institut für Mathematik, Freie Universität Berlin, Arnimallee 2, 14195, Berlin, Germany}

\author{Christian Schilling}
\email{c.schilling@physik.uni-muenchen.de}
\affiliation{Department of Physics, Arnold Sommerfeld Center for Theoretical Physics,
Ludwig-Maximilians-Universität München, Theresienstrasse 37, 80333 München, Germany}
\affiliation{Munich Center for Quantum Science and Technology (MCQST), Schellingstrasse 4, 80799 M\"unchen, Germany}

\date{\today}

\begin{abstract}
In [Phys.~Rev.~Lett.~127, 023001 (2021)] a reduced density matrix functional theory (RDMFT) has been proposed for calculating energies of selected eigenstates of interacting many-fermion systems. Here, we develop a solid foundation for this so-called $\bd{w}$-RDMFT and present the details of various derivations.
First, we explain how a generalization of the Ritz variational principle to ensemble states with fixed weights $\bd{w}$ in combination with the constrained search would lead to a universal functional of the one-particle reduced density matrix. To turn this into a viable functional theory, however, we also need to implement an exact convex relaxation. This general procedure includes Valone's pioneering work on ground state RDMFT as the special case $\bd{w}=(1,0,\ldots)$. Then, we work out in a comprehensive manner a methodology for deriving a compact description of the functional's domain. This leads to a hierarchy of generalized exclusion principle constraints which we illustrate in great detail.
By anticipating their future pivotal role in functional theories and to keep our work self-contained, several required concepts from convex analysis are introduced and discussed.
\end{abstract}

\maketitle

\section{Introduction}\label{sec:intro}
Based on a generalization\cite{Gilbert75} of the Hohenberg-Kohn theorem \cite{HK64}, one-particle reduced density matrix functional theory (RDMFT) has been proposed as an extension of density functional theory (DFT). By involving the full one-particle reduced density matrix (1RDM) $\g$, it immediately facilitates the exact description not only of the potential but also of the kinetic energy. Most importantly, RDMFT can describe strong correlation effects in a more directly manner since they manifest themselves in the form of fractional occupation numbers of the 1RDM. This distinguishes RDMFT compared to DFT as a more suitable approach to strongly correlated many-fermion systems  and explains why RDMFT has become an active research field in recent years \cite{C00,M07,TLMH15,PG16,SKB17,S18,MTP18,GUL18,SS19,BTNRR19,MPU19,GWK19,SBM19,C20b,C20a,G20,M21,Su21a}.
While the accuracy of ground state calculations compares favourably to those of DFT \cite{LM08}, no proper foundation for calculating energies of excited states exists yet.
For instance, a comprehensive justification of a fully dynamical RDMFT is lacking since it is unclear how the Runge-Gross theorem justifying time-dependent DFT can be extended to 1-RDMs and non-local potentials \cite{PG16}. Moreover,
resorting to the adiabatic approximation exploited through linear response techniques turns out to be rather challenging \cite{GGB12,PG16}.

All these points have recently motivated an extension of RDMFT to excited states \cite{SP21}. This so-called $\wb$-RDMFT is
inspired by the work \cite{GOK88a,GOK88b,GOK88c} in which Gross, Oliviera and Kohn extended ground state DFT to a (time-independent) DFT for targeting excited states. It is namely based on a generalization of the Ritz variational principle to ensemble states with spectrum $\wb$, used within the Levy-Lieb constraint search formalism \cite{LE79,Li83}, to establish a universal functional of the 1RDM.
The general hope is that this $\wb$-RDMFT overcomes the recent limitations of DFT to weakly correlated systems. At the same time, similarly to
time-independent DFT for excited states \cite{GOK88a,GOK88b,GOK88c} it shall circumvent the deficiencies of adiabatic time-dependent DFT  \cite{T79,GOK88b,F15B,YPBU17,GP17,SB18,GKP18,GP19,KF19,Fromager2020-DD,Loos2020-EDFA,GSP2020} and RDMFT \cite{Pernal07a,Pernal07b,Pernal07c,Appel2007,Giesbertz10}. In particular, $\bd w$-RDMFT will have a mathematically rigorous foundation in striking contrast to time-dependent RDMFT.

In the present work, we develop a solid foundation for $\bd{w}$-RDMFT and provide in addition detailed derivations and illustrations of all key results of the Letter \onlinecite{SP21}. For this and to keep our work self-contained, we first introduce and comprehensively explain in Sec.~\ref{sec:math} several required concepts from convex analysis. Since we expect that those concepts
will also play a pivotal role in the future development of functional theories, we dedicate to them a separate section.
Yet, the reader shall feel free to skip Sec.~\ref{sec:math} and come back to it whenever specific concepts from convex analysis are required for the discussion of RDMFT in the subsequent sections.
In Sec.~\ref{sec:introwrdmft}, we explain in great detail how a generalization of the Ritz variational principle in combination with the constrained search formalism leads to a universal functional of the one-particle reduced density matrix. To turn this into a viable functional theory, we apply in Sec.~\ref{sec:relaxw} an exact convex relaxation scheme, leading to the so-called $\wb$-RDMFT. In Sec.~\ref{sec:characterize}, we thoroughly explain how the functional's domain can be described in a compact manner.
This key achievement of our work leads to a hierarchical generalization of Pauli's famous exclusion principle which we comprehensively illustrate in Sec.~\ref{sec:examples} (see in particular Eqs.~\eqref{inr=2},\eqref{inr=3},\eqref{inr=4}, Fig.~\ref{fig:ew} and Tab.~\ref{tab:r}). Finally, in Sec.~\ref{sec:latticeDFT} we explain how those generalized exclusion principle constraints would also affect the analogous DFT.

\section{Key concepts from convex analysis}\label{sec:math}
In this section, we introduce and explain different basic concepts from convex analysis, such as
convex hulls, vector majorization, permutohedra, variational principles, exact convex relaxation and conjugation.
All these concepts will be essential for our work. They will namely allow us to circumvent or solve various conceptual problems
which have hampered the development of an RDMFT for excited states.
At the same time,  we also feel that a decent understanding and a more regular use of concepts from convex analysis would facilitate the development of a more profound foundation for functional theories in general. The more recent works \onlinecite{Li83,TKSETH12,KETH14,KLTH21} which are concerned with density functional theory provide further evidence for this expectation.
Yet, the reader shall feel free to skip the present section and come back to it whenever specific concepts from convex analysis are required in subsequent sections.

\subsection{Convex sets, convex hulls and lower convex envelops}\label{sec:conv}
A subset $S\subset \RR^d$ of the $d$-dimensional Euclidean space $\RR^d$ is called convex if for any two points $x,y \in S$ and any $0\leq q\leq 1$ the point $q x +(1-q)y$ also belongs to $S$.  The convex hull $\mbox{conv}(S)$ of a set $S$ is given by all convex combinations $\sum_{i}p_i x_i$ of finitely many points $x_i \in S$, where $0\leq p_i\leq 1$, $\sum_i p_i=1$. According to Carath\'eodory's theorem (see, e.g., textbook \onlinecite{R97}) one can restrict here to convex combinations of at most $d+1$ elements. Equivalently, as it is illustrated in Fig.~\ref{fig:convS}, $\mbox{conv}(S)$ is the smallest convex subset in $\RR^d$ which contains $S$ and thus $\mbox{conv}(S)= \bigcap \,\{\tilde{S}\subset \RR^d\,|\,S \subset \tilde{S},\, \tilde{S}\, \mbox{convex}\}$.

Let $f:S\rightarrow \RR$ be a real-valued function defined on a convex set $S$.
For the sake of completeness, we would like to stress that the following results and concepts would even hold for functions taking values in $\RR \cup \{\infty\}$.
\begin{figure}[htb]
\frame{\includegraphics[scale=0.47]{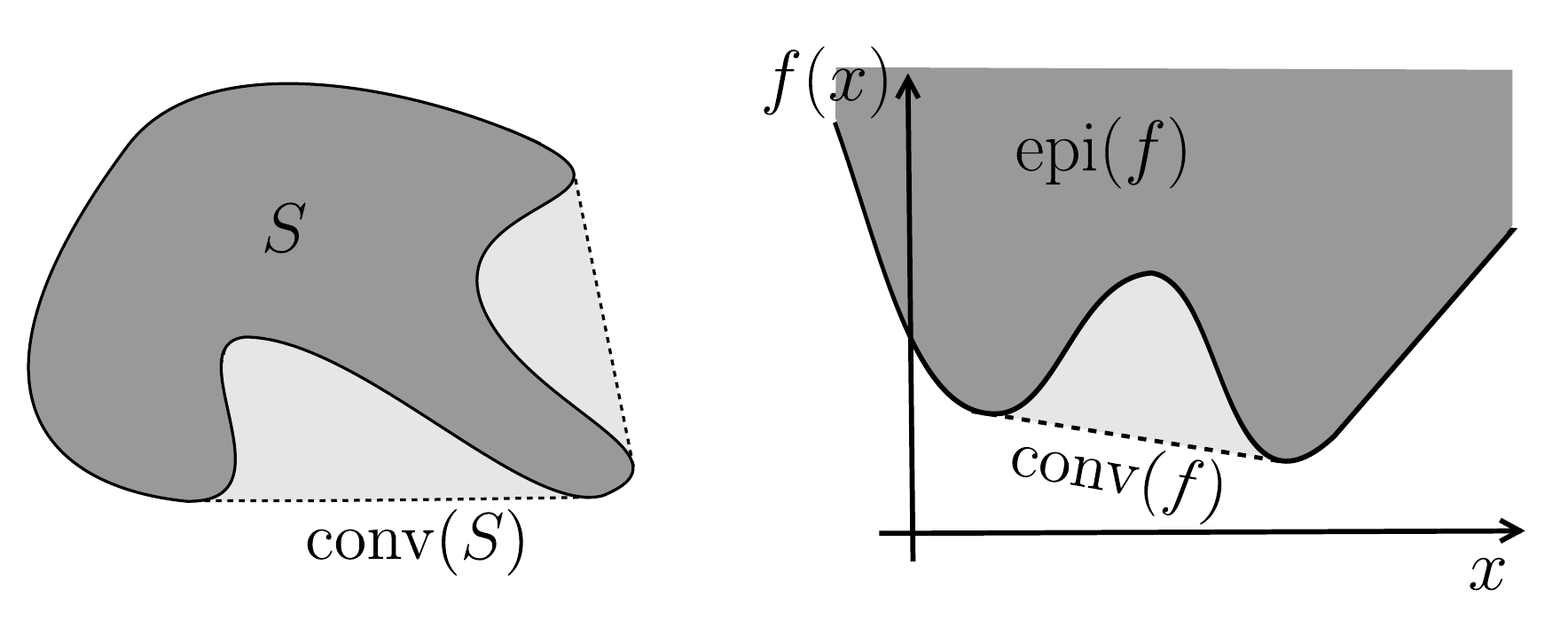}}
\caption{Left: Illustration of the convex hull $\mbox{conv}(S)$, i.e.~the smallest convex set which contains $S$. Right: Illustration of the lower convex envelop $\mbox{conv}(f)$ of a function $f$, defined through the convex hull (light gray) of the epigraph $\mbox{epi}(f)$ (gray).}
\label{fig:convS}
\end{figure}
The function $f$ is called convex if for any $x,y \in S$ and any $0 \leq q \leq 1$ we have
\begin{equation}\label{fconvex}
f\big(q x +(1-q)y\big) \leq q f(x)+(1-q)f(y)\,.
\end{equation}
Another useful concept used in the following is the one of the epigraph $\mbox{epi}(f)$ of a function $f$ (whose domain $S$ might be  non-convex),
\begin{equation}\label{epi}
\mbox{epi}(f) \equiv \{(x,y)\in \RR^d\times \RR\,|\,y \geq f(x)\}\,\subset \, \RR^{d+1}\,.
\end{equation}
The epigraph is the set of all points $(x,y)\in \RR^d\times \RR$ lying on or above the graph of $f$ (see also right side of Fig.~\ref{fig:convS}). Conversely, the epigraph of some function $f$ determines the function $f$ as well: For each value $x \in S$ one just needs to determine the smallest value $y$ under the condition $(x,y)\in \mbox{epi}(f)$, leading to $f(x)$. For more details we refer the reader to the textbook \onlinecite{R97}.

A subset $S \subset \RR^d$ is called closed if it contains all its limit points. To be more specific, for any sequence $(x_i)_{i=1}^\infty  \subset S$ converging to some $x \in \RR^d$, $\lim_{i\rightarrow \infty} x_i = x$, one has $x \in S$. The set $S$ is called compact if it is closed and bounded. Moreover, an extremal point of $S$  is a point $x \in S$ which cannot be written as a convex combination of other points in $S$. The significance of extremal points is emphasized by the Krein-Milman theorem. It states that any point $x$ in a compact set $S$ can be written as a (not necessarily unique) convex combination of extremal points of $S$.

The lower convex envelop $\mbox{conv}(f)$ of a function $f$ on $S$ is defined as the largest convex function $g$ on $S$ with $g \leq f$. Equivalently,
$\mbox{conv}(f)$ is the function corresponding to the convex hull of the epigraph of $f$, i.e., one exploits the relation $\mbox{epi}(\mbox{conv}(f))= \mbox{conv}(\mbox{epi}(f))$ (which also explains why one uses the same term, `\mbox{conv}' for both the convex hull and the lower convex envelop). In particular, the latter definition of the epigraph amounts to
\begin{equation}\label{convf}
\mbox{conv}(f)(x) = \inf \left\{\sum_i q_i f(x_i)\Big| \sum_i q_i x_i = x  \right\}\,,
\end{equation}
i.e., the infimum of $\sum_i q_i f(x_i)$ with respect to all possible convex decompositions $\sum_i q_i x_i$ of $x$. As a corollary of Carath\'eodory's theorem (see Corollary 17.1.5 in textbook \onlinecite{R97}) one could restrict such decompositions to at most $d+1$ elements. We illustrate the lower convex envelope of a function in the right panel of Fig.~\ref{fig:convS}. Since $f$ is not a convex function, also its epigraph (dark grey) is a non-convex set, whereas the epigraph of $\mathrm{conv}(f)$ (light and dark grey) represents a convex set.

\subsection{Vector majorization}\label{sec:major}
Let $\bd{v}, \bd{w}\in \RR^d$ be two vectors and denote by $\bd{v}^\downarrow, \bd{w}^\downarrow$ the vectors with the same entries, but sorted in descending order. We say that $\bd{w}$ majorizes $\bd{v}$ weakly if and only if for all $k=1,2,\ldots,d$:
\begin{equation}
v^\downarrow_1+\ldots + v^\downarrow_k \leq w^\downarrow_1+\ldots +w^\downarrow_k\,.
\end{equation}
If in addition the equality holds for $k=d$, $\sum_{j=1}^{d}v_j = \sum_{j=1}^{d}w_j$, we say that $\bd{w}$ majorizes $\bd{v}$, $\bd{v} \prec \bd{w}$. On the set $\RR^d$ the majorization $\prec$ defines a pre-order. Indeed, it is (i) transitive, $\bd{u} \prec \bd{v}$ and $\bd{v} \prec \bd{w}$ implies $\bd{u} \prec \bd{w}$ and (ii) reflexive,  $\bd{v} \prec \bd{v}$ for all $\bd{v}$.
Yet, it is not a total order, since not any two $\bd{u},\bd{v}$ are related and even not a partial order since it is not antisymmetric, i.e.,
$\bd{u} \prec \bd{v}$ and $\bd{v} \prec \bd{u}$ does not necessarily implies $\bd{u} = \bd{v}$ (but only $\bd{u}^\downarrow = \bd{v}^\downarrow$). Vector majorization will be a crucial tool to translate the RDMFT for excited states introduced in Sec.~\ref{sec:introwrdmft} into a practical method and reveals important relations between different sets of density matrices in Sec.~\ref{sec:relaxw}.

In 1923, Schur \cite{Schur1923} revealed an important connection between majorization and matrices (see also Ref.~\onlinecite{Horn54}).
\begin{thm}[Schur-Horn]\label{thm:Schur}
Let $A\in\RR^{d\times d}$ be a positive semi-definite Hermitian matrix, $\bd a\in\RR^d $ the vector of diagonal elements of $A$ and $\bd\lambda\in\RR^d$ the vector of eigenvalues of $A$. Then,
\begin{equation}
\bd a\prec \bd \lambda\,.
\end{equation}
\end{thm}
In Sec.~\ref{sec:latticeDFT}, we will use the Schur-Horn theorem to show that the constraints on the domain of the universal functional in $\wb$-ensemble RDMFT apply as well to excited state DFT for lattice systems.

A matrix $M \in \RR^{m\times d}$ with non-negative entries is called row stochastic (column stochastic) if for each row (column) all entries sum up to one. A square matrix $M \in \RR^{d\times d}$ is called doubly stochastic (or bistochastic) if it is at the same time row and column stochastic. The following three theorems \cite{HLP53,AU82,B46,vN53} will be crucial for our work, especially for the development of RDMFT for excited states in Sec.~\ref{sec:introwrdmft} and Sec.~\ref{sec:relaxw}.
\begin{thm}[Hardy, Littlewood, P\'olya]\label{thm:HLP}
Let $\wb, \wb'\in \RR^d$. Then,
\begin{eqnarray}
\wb' \prec \wb \quad &\Leftrightarrow & \quad \exists M \in \RR^{d\times d}\,\mbox{double stochastic}\!: \nonumber \\
&& \quad \wb' = M \wb\,.
\end{eqnarray}
\end{thm}
\begin{thm}[Uhlmann]\label{thm:Uhl}
Let $\G$ and $\G'$ be two density operators on a $D$-dimensional Hilbert space.
Then there exist unitary operators $U_i$ and weights $0\leq p_i\leq 1$, $\sum_{i}p_i=1$, such that
\begin{equation}
\G' = \sum_{i=1}^D p_i U_i \G U_i^\dagger
\end{equation}
if and only if the spectrum of $\G$ majorizes the spectrum of $\G'$, $\mbox{spec}(\G') \prec \mbox{spec}(\G)$.
\end{thm}
For the sake of completeness, we recall that the density operators $\Gamma$ are by definition hermitian, positive semidefinite and normalized to one, i.e.~$\Tr[\Gamma]=1$.
Uhlmann's theorem is the quantum generalization of Theorem \ref{thm:HLP}. Indeed, the classical case is part of the more general quantum context and would correspond to two density operators $\G$ and $\G'$ being diagonal with respect to the same reference basis. In that case,  Theorem \ref{thm:HLP} could be applied to the vector of diagonal entries of $\G$ and $\G'$, leading to the same predictions as Uhlmann's theorem.

\begin{thm}[Birkhoff, von Neumann]\label{thm:Birk}
The extremal points of the convex set of doubly stochastic matrices in $\RR^{d\times d}$ are given by the $d!$ many permutation matrices.
\end{thm}
We would like to recall that a permutation matrix is a square matrix which has in each row and column exactly one entry $1$ and $d-1$ entries $0$. Consequently, it is in particular also doubly stochastic.

\subsection{Permutohedra and polytopes in general}\label{sec:polytopes}
A quite specific and important class of (compact) convex sets is the one of convex polytopes. A convex polytope $P \subset \RR^d$ is a subset of a Euclidean space which is given as the convex hull of finitely many extremal points. This family of extremal points (vertices) defines the polytope in its so-called $v$-representation. Equivalently, a polytope can be defined as the intersection of finitely many halfspaces, i.e., a finite family of linear constraints ($h$-representation). Translating the $v$-representation of a polytope into its $h$-representation is computationally quite demanding for general polytopes in higher dimensions with a large number of vertices (see, e.g., Ref.~\onlinecite{AF92}). For each polytope, there exists
a minimal $h$-representations which is unique up to rescaling of the linear inequalities. The inequalities belonging to the minimal $h$-representation are called facet-defining.

Let $\wb\in \RR^d$ be a vector. We define the corresponding \emph{permutohedron} $P_{\wb}$ as the convex hull of all permuted versions of the vector $\wb$,
\begin{equation}
  P_{\wb} \equiv \mbox{conv}\left(\{\pi(\wb)\}_{\pi \in \mathcal{S}^d}\right)\,.
\end{equation}
Here, $\mathcal{S}^d$ denotes the group of permutations $\pi$ of $d$ elements (i.e., $\pi$ permutes the $d$ entries of the vector $\wb$).
Clearly, the permutohedron is permutation-invariant since
$\pi(P_{\wb})=P_{\pi(\wb)} = P_{\wb}$ for any $\pi \in \mathcal{S}^d$.

Notice, the permutohedron $P_{\wb}\subset \RR^d$ lies in the hyperplane of vectors $\bd{v}\in \RR^d$ fulfilling $\sum_{i=1}^{d}v_i=\sum_{i=1}^{d}w_i$. Consequently, it has dimension of at most $d-1$.

A crucial theorem states \cite{R52}
\begin{thm}[Rado]\label{thm:Rado}
Let $\wb \in \RR^d$. Then,
\begin{equation}
\bd{v} \in P_{\wb} \quad \Leftrightarrow \quad \bd{v} \prec \wb\,.
\end{equation}
\end{thm}
This will be a key tool for characterizing in Sec.~\ref{sec:vtoh} the domain of the universal functional in $\wb$-RDMFT which will take effectively the form of either a permutohedron or a generalization thereof. 

\subsection{Variational principles}\label{sec:variational}
Let $H$ be a Hermitian operator (e.g., a Hamiltonian) on a $D$-dimensional Hilbert space $\H$ with eigenvalues $E_1 \leq E_2 \leq \ldots \leq E_D$. The Ky Fan variational principle \cite{F49} then states that the sum of the $k$ smallest eigenvalues can be expressed in terms of a variational principle,
\begin{equation}
\sum_{j=1}^k E_j = \min_{V \leq \H, \mathrm{dim}(V)=k}\Tr[P_V H]\,.
\end{equation}
Here, $P_V$ denotes the orthogonal projection operator onto the $k$-dimensional subspace $V$ and the minimizer is given by the span of the eigenstates of the $k$-smallest eigenvalues, $E_1,\ldots,E_k$. Replacing the minimization by a maximization would yield the sum of the $k$ largest eigenvalues. Clearly, Ky Fan's principle includes in the form of $k=1$ the Rayleigh-Ritz variational principle. Moreover, it is worth noticing that the Ky Fan variational principle has been rediscovered and used in Ref.~\onlinecite{T79} to establish an equiensemble DFT.

A straightforward extension of Ky Fan's principle is the variational principle introduced and used in the formulation of $\wb$-ensemble density functional theory \cite{GOK88a,GOK88b,GOK88c}:
\begin{thm}[Gross, Oliviera, Kohn]\label{thm:GOK}
Let $H$ be a Hermitian operator (typically an $N$-fermion Hamiltonian) on a $D$-dimensional Hilbert space $\H$ with increasingly ordered eigenvalues $E_1 \leq E_2 \leq \ldots \leq E_D$ and $\wb \in \RR^D$ with decreasingly ordered entries. Then, by denoting the family of density operators with spectrum $\wb$ by $\Ew$ the following variational principle holds
\begin{equation}
E_{\wb} \equiv \sum_{j=1}^{D} w_j E_j = \min_{\G \in \Ew} \Tr[\G H ]\,.
\end{equation}
In particular, the minimizer $\G_{\wb}$ on the right-hand side is given by the state $\sum_{j=1}^{D} w_j \ket{\Psi_j}\!\bra{\Psi_j}$, where $\ket{\Psi_j}$ is the eigenstate of $H$ corresponding to the eigenvalue $E_j$.
\end{thm}

It is exactly this variational principle due to Gross, Oliviera and Kohn which together with the constrained search formalism allows one to develop an $\wb$-ensemble density or density-matrix functional theory, targeting the excited states.

\subsection{Legendre-Fenchel transformation}\label{sec:conjug}
Let $f$ be a function defined on the Euclidean space $\RR^d$ which takes values in $\RR \cup\{\infty\}$. Allowing $f$ to take infinite values ``has the advantage that technical nuisances about effective domains can be suppressed almost entirely''\cite{R97} in the context of convex analysis. Nonetheless, we assume $f$ to be proper, $f \not \equiv \infty$ (i.e., at least for some $x \in \RR^d$ it takes finite values). The conjugate $f^*$ (also called Legendre-Fenchel transform) of $f$ is defined as
\begin{equation}
f^*(y)=\sup_{x\in \mathbb{R}^d}\big[\langle y,x \rangle-f(x)\big]\,.
\end{equation}
Here, $\langle \cdot,\cdot \rangle$ denotes the standard inner product on $\RR^d$ which equips $\RR^d$ with a notion of geometry (``distances and angles''). The function $f^*$ is defined on $\RR^d$ in the sense that it is a map from $\RR^d$ to $\RR \cup\{\infty\}$ (like $f$).

One of the crucial theorems in convex analysis relates $f$ with its biconjugation:
\begin{thm}[biconjugation]\label{thm:biconjug}
Let $f: \RR^d \rightarrow \RR\cup\{\infty\}$ be proper. Then,
\begin{equation}
f^{**} = \mathrm{cl}\left(\mathrm{conv}(f)\right)\,,
\end{equation}
where $\mbox{cl}$ denotes the closure operation.
\end{thm}
Let us briefly explain the definition of the closure operator: Given the epigraph of a function $f$ (recall Eq.~\eqref{epi}), one can extract the function $f$ from it. The function $\mbox{cl}(f)$ is defined as the function corresponding to the closure of the subset $\mbox{epi}(f)\subset \RR^{d+1}$. This in particular means that the closure operation only affects the function at those points where it is not continuous.

\subsection{Relaxing minimizations to convex minimizations}\label{sec:relax}
One of the main achievements of our work will be to circumvent the highly involved $\wb$-ensemble $N$-representability constraints (as introduced in Sec.~\ref{sec:introwrdmft}) by relaxing the universal functional to a larger domain. To explain the underlying procedure, let us consider a
function $f: S \rightarrow \RR$ defined on a subset $S\subset \RR^d$ for which we would like to determine its infimum,
\begin{equation}\label{inf1}
  f_0\equiv \inf_{x \in S} f(x)\,.
\end{equation}
Whenever $S$ is compact (i.e., it is a bounded and closed subset of $\RR^d$) and $f$ continuous the infimum is attained. In that case we can replace ``inf'' in Eq.~\eqref{inf1}  by ``min''.

We could extend the minimization in Eq.~\eqref{inf1} to the whole Euclidean space $\RR^d$ by extending the function $f$ according to
\begin{equation}\label{fext}
  f_{\mathrm{ext}}(x) \equiv \Big\{\begin{array}{lr}
        f(x), & \text{for } x \in S\\
        \infty, & \text{for } x \not \in S       \end{array}\,,
\end{equation}
without affecting the infimum,
\begin{equation}\label{inf2}
  f_0 = \inf_{x \in \RR^d} f_{\mathrm{ext}}(x)\,.
\end{equation}
Furthermore, replacing $f_{\mathrm{ext}}$ by its lower convex envelop $\mbox{conv}(f_{\mathrm{ext}})$, i.e., the largest convex function $g$ on $\RR^d$ with $g \leq f_{\mathrm{ext}}$ (recall Sec.~\ref{sec:conv}) would not change the infimum,
\begin{equation}\label{inf3}
  f_0= \inf_{x \in \RR^d} \big[\mbox{conv} (f_{\mathrm{ext}})(x)\big]\,.
\end{equation}
The conclusion of those considerations is that we can restrict the class of minimization problems of the type \eqref{inf1} to minimization problems involving only convex functions on convex domains, i.e., for any $f$ on some $S\subset \RR^d$ we relax the problem to $\xbar{f}\equiv \mbox{conv}(f)$ on $\xbar{S}\equiv \mbox{conv}(S)$. This so-called \emph{exact convex relaxation} \cite{R97} has two advantages: The minima of $\xbar{f}$ are always global ones and the description of convex sets $\xbar{S}$ is typically easier than the one of non-convex sets. To be more specific, in the context of $\wb$-ensemble RDMFT, this will relax the highly involved $\wb$-ensemble $N$-representability constraints to \emph{easy-to-calculate} Pauli-like constraints.

\subsection{Characterization of compact convex sets: Duality correspondence}\label{sec:dual}
Our work involves the maximization (or minimization) of linear functions $L_k(\cdot)\equiv \langle \cdot ,k\rangle$ over a (convex compact) set $S\subset \RR^d$. It is instructive to understand how the maximum would change as function of $k \in \RR^d$. Formally, this means to study the so-called support function
\begin{equation}\label{supp}
\chi_S^\ast(k) \equiv \sup \big(\{\langle x ,k\rangle|x \in S\}\big)\,,
\end{equation}
where `sup' stands for supremum.
Clearly, this covers also the case of minimizations:
\begin{equation}\label{suppinf}
\inf \big(\{\langle x ,k\rangle|x \in S\}\big) = - \chi_S^\ast(-k)\,.
\end{equation}
The reason for denoting the support function by $\chi_S^\ast$ is that it is the conjugate of the indicator function which is defined as
\begin{equation}\label{ind}
\chi_S(x)   \equiv \Big\{\begin{array}{lr}
        0, & \text{for } x \in S\\
        \infty, & x \not \in S       \end{array}\,.
\end{equation}
Indeed, calculating the Legendre-Fenchel transform/conjugate of $\chi_S$ leads to
\begin{eqnarray}
\chi_S^\ast(k) &\equiv& \sup_{x \in \RR^d}\left[\langle x, k \rangle - \chi_S(x)\right] \nonumber \\
&=& \sup_{x \in S} \langle x, k \rangle \,.
\end{eqnarray}

For convex compact sets even a stronger relation holds (see, e.g., textbook \onlinecite{R97}):
\begin{thm}[duality correspondence]\label{thm:dual}
For any convex, compact set $S\subset \RR^d$, the support function $\chi_S$ and the indicator function $\chi_S^\ast$ are conjugate to each other
\begin{equation}
\left(\chi_S^\ast\right)^\ast= \left(\chi_S\right)^{\ast\ast}=\chi_S\,.
\end{equation}
This means that not only the support function is the conjugate of the indicator function but also the indicator function follows as the conjugate of the support function, establishing a one-to-one correspondence between them.
\end{thm}
It will be exactly this duality correspondence which allows us in combination with some additional tools to characterize below the domain of the universal excited state functional. Remarkably, this will then lead to a hierarchical generalization of Pauli's famous exclusion principle.

\section{Introducing $\wb$-ensemble reduced density-matrix functional theory}\label{sec:introwrdmft}
\subsection{The fundamental motivation of RDMFT}\label{sec:RDMFTmotiv}
To first motivate RDMFT from a more general point of view, we recall that each scientific field restricts to quantum systems which are typically characterized by the same (possibly effective) pair interaction $W$. Just to name the most important ones in non-relativistic quantum mechanics, these are the Coulomb interaction in quantum chemistry, the Hubbard onsite interaction in discretized  models in solid state physics and hard-core interaction in quantum optics, particularly in the field of ultracold atoms. Consequently, the family of relevant Hamiltonians \emph{in each} of those fields is parameterized by the particle number  $N$, the one-particle Hamiltonian $h$ (i.e., the kinetic energy/hopping operator and the external potential) and the coupling constant. Restricting to systems of fixed particle number $N$ and absorbing the coupling constant into the one-particle Hamiltonian leads to Hamiltonians of the form
\begin{equation}\label{ham}
H(h) = h +  W\,.
\end{equation}
Here, the fixed interaction $W$ could in principle be even any $p$-particle interaction with $2\leq p \leq N$. We restrict ourselves to finite $N$-fermion Hilbert spaces $\mathcal{H}_N \equiv \wedge^N\mathcal{H}_1$, where $\mathcal{H}_1$ is the underlying finite $d$-dimensional one-particle Hilbert space. In practice, $\mathcal{H}_1$ follows by choosing a truncated finite basis set.

Since the class \eqref{ham} of Hamiltonians is fully parameterized by $h$, various properties of any eigenstate of $H(h)$ depend on $h$ only.
This is nothing else than a trivial consequence of a scientific community's decision to restrict the electron correlation problem to a specific class of Hamiltonians of the form \eqref{ham}. Contrary to the common perception in ground state DFT or RDMFT, this is therefore not an astonishing implication of the Hohenberg-Kohn \cite{HK64}  or Gilbert theorem \cite{G75}. Moreover, it is not surprising that the ground state problem (for a fixed interaction $W$) will involve the one-particle reduced density matrix (1RDM) $\g$ only since the latter is the conjugate variable of the one-particle Hamiltonian $h$.
The corresponding ground state theory is called reduced density matrix functional theory (RDMFT). From a historic point of view, it is based on Gilbert's generalization \cite{G75} of the famous Hohenberg-Kohn theorem \cite{HK64} to non-local external ``potentials''.

\subsection{RDMFT for ground states}\label{sec:RDMFTgs}
We briefly recall ground state RDMFT. For this, we introduce the set $\mathcal{P}^N$ of {\em pure} $N$-fermion states $\G\equiv \ket{\Psi}\!\bra{\Psi}$. By referring to the Ritz variational principle the ground state energy $E(h)$ of $H(h)$ can be expressed as  \cite{LE79,Li83},
\begin{eqnarray}\label{Levy}
  E(h) &=& \min_{\G \in \mathcal{P}^N} \mbox{Tr}_N[(h+W)\G] \nonumber \\
 &=&  \min_{\gamma\in \mathcal{P}^1_N}\Big[\mbox{Tr}_1[h\gamma]+\F(\g)\Big]\,,
\end{eqnarray}
where
\begin{equation}
\mathcal{P}^1_N\equiv N \mbox{Tr}_{N-1}[\mathcal{P}^N]\,.
\end{equation}
The underlying \emph{pure} functional
\begin{eqnarray}\label{Fp}
\F(\g) \equiv \min_{\mathcal{P}^N\ni \G \mapsto \gamma}\mbox{Tr}_N[W \G]
\end{eqnarray}
is universal in the sense that it depends only on the fixed interaction $W$ but not on the one-particle Hamiltonian $h$. In \eqref{Fp} the expression $\G \mapsto \gamma$ means to minimize only with respect to those $N$-fermion pure states $\G \in \mathcal{P}^N$ which map to the given 1RDM $\g$.
This pure RDMFT due to Levy \eqref{Levy} based on \eqref{Fp} is, however, not practical. The reason for this is that the domain $\mathcal{P}^1_N$ of pure $N$-representable 1RDMs $\g$ is too involved. Only recently, a formal solution of the one-body pure $N$-representability problem has been found \cite{KL06,AK08}. Yet, the corresponding constraints defining $\mathcal{P}^1_N$  (also called generalized Pauli constraints) could be calculated so far only for systems of up to five electrons and eleven spin-orbitals \cite{BD72,KL06,AK08,S18atoms}. Since also the number of those constraints increases drastically as function of the particle number $N$ and basis set size $d$ one may need to rely in quantum chemical applications on approximate description of $\mathcal{P}^1_N$ \cite{MT17}.

A much more promising and simpler idea due to Valone \cite{V80} is to relax the constrained search in \eqref{Levy},\eqref{Fp}
to the convex set of \emph{all} ensemble states
\begin{equation}\label{EN}
\mathcal{E}^N \equiv \big\{\Gamma:\mathcal{H}_N \rightarrow \mathcal{H}_N\,| \,\mbox{linear}, \Gamma \geq 0, \mbox{Tr}[\Gamma]=1\big\}\,.
\end{equation}
The reader shall recall that $\Gamma \geq 0$ means that $\Gamma$ is positive semidefinite, i.e., $\bra{\Psi}\Gamma \ket{\Psi} \geq 0$ for all $\ket{\Psi} \in \wedge^N \mathcal{H}_1$. Moreover, in case of finite dimensional Hilbert spaces, hermiticity follows from positive semidefiniteness.
In analogy to \eqref{Levy} and \eqref{Fp}, this then leads to a universal \emph{ensemble} functional
\begin{eqnarray}\label{Fe}
\xbar{\F}(\g)  &\equiv&  \min_{\E \ni \G \mapsto \g}\mbox{Tr}_N[W \G ]\,,
\end{eqnarray}
which is defined on the larger domain
\begin{equation}\label{E1}
\e \equiv  N \mbox{Tr}_{N-1}[\E]
\end{equation}
of ensemble $N$-representable 1RDMs. The set $\e$ is convex (in contrast to $\mathcal{P}^1_N$) which follows from the linearity of the partial trace map $\mbox{Tr}_{N-1}[\cdot]$ and the convexity of $\mathcal{E}^N$. Since the set $\e$ is described by the simple Pauli exclusion principle constraints \cite{K60,C63}, restricting various eigenvalues $\lambda_i$ of $\g$ according to
\begin{equation}\label{PC}
0\leq \lambda_i \leq 1,
\end{equation}
Valone's work \cite{V80} from 1980 has marked the starting point of a \emph{practically feasible} ground state RDMFT.

Let us finally comment on a crucial relation between pure and ensemble ground state RDMFT. First,
the compact convex set $\mathcal{E}^1_N$ is equal to the convex hull of $\mathcal{P}^1_N$ (recall the definition of the convex hull in Sec.~\ref{sec:conv}),
\begin{equation}\label{E1P1}
\mathcal{E}^1_N = \mathrm{conv}(\mathcal{P}^1_N)\,.
\end{equation}
This follows directly from $\mathcal{E}^N = \mathrm{conv}(\mathcal{P}^N)$ (see, e.g., Theorem 2.6 in the book \onlinecite{CY00}) in combination with the linearity of the partial trace $\mbox{Tr}_{N-1}[\cdot]$.
Moreover, the universal ground state functionals $\mathcal{F}$ and $\xbar{\mathcal{F}}$ are related to the ground state energy $E(h)$ through the Legendre-Fenchel transformation discussed in Sec.~\ref{sec:conjug} (see also Ref.~\onlinecite{S18} for further details). Consequently, the concept  of biconjugation (see Theorem \ref{thm:biconjug}) implies \cite{S18}
\begin{thm}[Schilling]\label{thm:CS} The pure ($\mathcal{F}$) and ensemble functional (\,$\xbar{\F}$) in ground state RDMFT are related according to
\begin{equation}
\xbar{\F}=\mathrm{conv}(\mathcal{F})\,,
\end{equation}
i.e., \,$\xbar{\F}$ is the lower convex envelope (see Sec.~\ref{sec:conv}) of $\mathcal{F}$.
\end{thm}
According to Theorem \ref{thm:CS}, Valone's approach is nothing else than a convex relaxation (recall Sec.~\ref{sec:relax}) of Levy's pure RDMFT. This observation has far-reaching consequences since it will allow us to propose in the next section a practically feasible RDMFT for excited states. For this, we employ the GOK variational principle, Theorem \ref{thm:GOK}, and then turn this into a viable $\wb$-ensemble RDMFT by applying convex relaxation.

\subsection{$\wb$-ensemble RDMFT}\label{sec:wRDMFT}
In analogy to ensemble DFT (GOK-DFT) for excited states by Gross, Oliviera and Kohn \cite{GOK88a,GOK88b}
we apply the variational principle described by Theorem \ref{thm:GOK} to the Hamiltonian $H(h)$ \eqref{ham}. For this, we denote
the increasingly ordered eigenvalues of $H(h)$ by $E_1(h)\leq E_2(h)\leq \ldots \leq E_D(h)$ and recall that $\wb$ is a tuple of decreasingly ordered weights $w_1 \geq w_2 \geq \ldots \geq w_D\geq 0$, $\sum_{j=1}^D w_j=1$, where $D \equiv \mbox{dim}(\HN)=\binom{d}{N}$.

We first observe that having access to the function
\begin{eqnarray}\label{Ew}
E_{\wb}(h) &\equiv& \sum_{j=1}^{D} w_j E_j(h) \nonumber \\
&=& E_1(h)+ \sum_{j=2}^{D} w_j \big[E_j(h) - E_1(h)\big]
\end{eqnarray}
would obviously allow us to determine all the eigenenergies and the gaps between them:
\begin{equation}\label{Ewext}
\frac{\partial E_{\wb}(h)}{\partial w_j}=E_j(h)-E_1(h)\quad,\,j=2,\ldots,D\,.
\end{equation}
The partial derivative in Eq.~\eqref{Ewext} is defined in this context with respect to the second line of \eqref{Ew}, i.e., after having used the normalization of $\wb$ to unity to replace $w_1$.
In particular, these derivatives are independent of $\wb$ since $E_{\wb}$ depends only linearly on the weights. Equivalently, one could therefore extract the energy gaps $E_j(h)-E_1(h)$ by considering an appropriate gradient triangle. The individual excited state energies $E_j(h)$ would follow immediately by using $E_1(h)= E_{(1,0,\ldots)}(h)$.

The GOK variational principle, Theorem \ref{thm:GOK}, together with the constrained search formalism leads directly to some (rather impractical) form of $\wb$-ensemble RDMFT. To work this out, we first define the underlying spaces of density operators.
These are on the $N$-fermion level
\begin{equation}\label{ENw}
\Ew \equiv \big\{\Gamma \in \mathcal{E}^N|\, \mbox{spec}^\downarrow \left(\Gamma\right)=\wb\big\}\,,
\end{equation}
the set of all $N$-fermion density operators with spectrum $\wb$.
Furthermore, tracing out all except one fermion then leads to
\begin{equation}\label{E1w}
\ew \equiv N \mbox{Tr}_{N-1}[\Ew]\,.
\end{equation}
We refer to $\gamma \in \ew$ as being \emph{$\wb$-ensemble $N$-representable}. For the following considerations, we would like to emphasize as well that both sets  $\Ew$ and $\ew$ are not convex: A convex sum of two density operators with spectrum $\wb$ has typically a spectrum different than $\wb$. Due to the linearity of the partial trace map $N\Tr_{N-1}[\cdot]$ the non-convexity of $\ew$ is therefore not surprising. More rigorously, this can be deduced from Klyachko's work on the mixed single-body quantum marginal problem \cite{KL04,KL09,AK08}.

To define a universal functional for targeting energies of excited states, we express $E_{\wb}(h)$ according to Theorem \ref{thm:GOK} as a variational principle. By referring to the class of Hamiltonians of the form \eqref{ham} (i.e., with a fixed interaction $W$, fixed particle number $N$ and a variable one-particle Hamiltonian $h$) we directly find
\begin{eqnarray}\label{Levy1}
  E_{\wb}(h) &=& \min_{\Gamma\in \Ew} \mbox{Tr}_N[(h+W)\Gamma] \nonumber \\
 &=&  \min_{\gamma\in \ew}\Big[\mbox{Tr}_1[h\gamma]+\Fw(\gamma)\Big]\,,
\end{eqnarray}
where
\begin{equation}\label{Fw}
\Fw(\gamma) \equiv \min_{\Ew\ni\Gamma\mapsto \gamma}\mbox{Tr}_N[W\Gamma]\,.
\end{equation}
Just to reiterate, our approach is not based on a generalization of the Hohenberg-Kohn or Gilbert theorem to $\wb$-ensembles but exploits the constrained search formalism \eqref{Levy1} instead. As a consequence, we circumvent the corresponding $v$-representability problem by extending the functional's domain from 1RDMs of $\wb$-minimizers to all $\wb$-ensemble $N$-representable 1RDMs, i.e., to the set $\ew$. Furthermore, because of the latter, we can ignore any subtleties due to possible degeneracies within the energy spectrum.

The domain of the $\wb$-ensemble functional $\mathcal{F}_{\wb}(\gamma)$ defined in Eq.~\eqref{Fw} contains all $\wb$-ensemble $N$-representable 1RDMs $\gamma\in \ew$. As we will explain in the next section, it would be a hopeless endeavor to determine an efficient description of $\ew$. Instead, we need to apply a convex relaxation scheme to circumvent those conceptual difficulties,
eventually leading to a  \emph{viable} functional theory. Before that, we conclude the present section by stating
\begin{thm}\label{thm:FwtoFpe}
The relation of $\Fw$ to the pure $(\F)$ and ensemble $(\,\xbar{\F})$ ground state functional (recall \eqref{Fp},\eqref{Fe}) is given by
\begin{eqnarray}
\F(\g) &= & \mathcal{F}_{\!\wb_0}(\g)\,, \,\,\quad \quad \forall \g \in\mathcal{E}^1_N(\wb_0)\equiv \mathcal{P}^1_N \label{Fprel}\\
\xbar{\F}(\g) &=& \min_{\wb} \Fw(\g)\,, \quad \forall \g \in \mathcal{E}^1_N\,, \label{Ferel}
\end{eqnarray}
where
\begin{equation}
\wb_0\equiv (1,0,0,\ldots)
\end{equation}
and we extended here $\Fw$ to $\e$ by setting $\Fw(\g) = \infty$ for all $\g \not \in \ew$.
\end{thm}
\begin{proof}
Equation \eqref{Fprel} follows directly from the definition of $\F$ and  $\mathcal{F}_{\!\wb_0}$, and by recalling that
\begin{equation}
\mathcal{P}^N = \mathcal{E}^N(\wb_0)\,.
\end{equation}
To prove \eqref{Ferel}, let $\g \in \mathcal{E}^1_N$. According to the constrained search formalism we have
\begin{eqnarray}
\xbar{\F}(\g) &\equiv& \min_{\E \ni\G \mapsto \g}\Tr[W \G] \nonumber \\
&=& \min_{\wb}\min_{\Ew \ni \G \mapsto \g}\Tr[W \G] \nonumber \\
&\equiv& \min_{\wb} \Fw(\g)\,,
\end{eqnarray}
where we set  $\min_{\Ew \ni \G \mapsto \g}\Tr[W \G] \equiv \infty$  whenever \mbox{$\{ \G \in \Ew \,|\, \G \mapsto \g \}= \emptyset$}.
\end{proof}

\section{Relaxation of $\wb$-ensemble RDMFT}\label{sec:relaxw}
The RDMFT as introduced in the previous section based on Eqs.~\eqref{Levy1},\eqref{Fw} would in particular require an efficient
description of the domain $\ew$. To elaborate on this, we first observe that the non-convex set $\ew$ is described by purely spectral constraints, i.e., we would need to know only the eigenvalues (\emph{natural occupation numbers}) of a 1RDM $\g$ to decide whether it belongs to $\ew$ or not. This is due to the following unitary equivalence
\begin{equation}\label{unitequiv}
\g \in \ew \quad \Rightarrow \quad u \g u^\dagger \in \ew\,,
\end{equation}
valid for all unitaries $u$ on the one-particle Hilbert space.
Indeed, when $\Ew \ni \G \mapsto \g$ then $u^{\otimes^N} \G (u^\dagger)^{\otimes^N} \mapsto u \g u^\dagger$, where $u^{\otimes^N} \G \, (u^\dagger)^{\otimes^N} \in \Ew$ since $u^{\otimes^N}$ is a unitary on the $N$-fermion Hilbert space (and therefore does not change the spectrum of $\G$). While the unitary equivalence simplifies the description of $\ew$ considerably, this task is unfortunately still quite challenging. Actually, there exists from a more abstract mathematical point of view a general solution to that problem in the form of Klyachko's seminal work on the mixed single-body quantum marginal problem \cite{KL04,KL09,AK08}. Yet from a practical point of view the recent algorithms are not efficient enough to determine the spectral set, $\mbox{spec}^\downarrow\left(\ew\right)$, which in turn takes the form of a polytope. Even if we knew that polytope in its hyperplane representation for reasonably large active spaces sizes, it might be difficult to take into account all those $\wb$-ensemble generalized Pauli constraints due to their sheer number. It is therefore one of the crucial contributions of our present work, to develop and propose a procedure which allows us to circumvent all those fundamental mathematical and computational problems. The underlying idea can be seen as a generalization of Valone's work \cite{V80} which allowed one to relax in ground state RDMFT the underlying domain $\p$ of pure $N$-representable 1RDMs to the convex set $\e$ of ensemble $N$-representable 1RDMs. Valone's work was based on the fact that the purity constraint on the $N$-fermion quantum state in the variational ground state problem could be skipped. Our more general perspective allows us to identify Valone's idea as a rather natural relaxation procedure with similarly striking impact on $\wb$-ensemble RDMFT: For any $\wb$ (including $\wb= \wb_0$, i.e., ground state RDMFT), we can relax, as explained in Sec.~\ref{sec:relax}, the problem of minimizing a non-convex functional $\Fw$ on a highly involved non-convex set $\ew$ to a convex minimization problem:
\begin{thm}\label{thm:relax}
For the class of Hamiltonians $H(h)$ \eqref{ham} and any choice of decreasingly ordered ensemble weights $\wb$ we can determine the energy $E_{\wb}$ \eqref{Ew} and the 1RDM of the corresponding $\wb$-minimizer $\G_{\wb}$ by minimizing the relaxed energy functional $\langle h, \g\rangle + \Fbw(\g)$,
\begin{equation}
E_{\wb}= \min_{\g \in \ebwS} \left[\langle h, \g\rangle + \Fbw(\g) \right]\,,
\end{equation}
over its domain
\begin{equation}\label{ebw}
\ebw \equiv \mathrm{conv}(\ew)\,.
\end{equation}
The relaxed functional $\Fbw$ is given as the lower convex envelop of the original functional $\Fw$ (as explained in Sec.~\ref{sec:relax}) and it is universal in the sense that it does not depend on the variable one-particle Hamiltonian $h$.
\end{thm}
We would like to stress again that relaxing $\wb$-ensemble RDMFT to such a convex minimization problem has crucial advantages. First, it will turn out that the set $\ebw$ can be determined for \emph{arbitrary} system sizes $(N,d)$ and the constraints depend \emph{de facto} only on the number $r$ of non-vanishing weights $w_j$. This means that in striking contrast to the pure $N$-representability conditions (generalized Pauli constraints) they do effectively not depend on the particle number and the basis set size $d$. Second, the set $\ebw$ and the functional $\Fbw$ are both convex. This implies that any minimum obtained from the minimization of $\Fbw$ will automatically be the global one.

In analogy to the ensemble ground state functional $\xbar{\F}$ proposed by Valone, also our relaxed functional $\Fbw \equiv \mbox{conv}(\Fw)$ can be derived through the constrained search formalism
\begin{thm}\label{Fbwmajor}
The relaxed $\wb$-ensemble functional $\Fbw$ follows as
\begin{equation}\label{Fbwmajor}
\Fbw(\g) = \min_{\EbwS \ni \G \mapsto \g} \mathrm{Tr}[\G W]\,,
\end{equation}
and in addition we have (with the majorization $\prec$ defined in Sec.~\ref{sec:major})
\begin{eqnarray}
\Ebw  &=& \!\bigcup_{\wb'\prec\,\wb}\!\mathcal{E}^N(\wb')\equiv \{\G \in \E| \mathrm{spec}(\G) \prec \wb\} \label{Ebwmajor}  \\
\ebw &=& N\Tr_{N-1}[\Ebw]= \bigcup_{\wb'\prec\,\wb} \mathcal{E}^1_N(\wb')\,.\label{ebwmajor}
\end{eqnarray}
\end{thm}
\begin{proof}
Equation \eqref{Ebwmajor} follows directly from Uhlmann's theorem (see Theorem \ref{thm:Uhl}). To explain this, we first recall from Eq.~\eqref{ENw} that the set $\Ew$ is fully characterized through the fixed spectrum $\wb$. Since unitary transformations do not change the spectrum, $\Ew$ can be parameterized as the family of all unitary conjugations of some arbitrary $\G \in \Ew$, $\Ew = \{U\G U^\dagger\}$. Accordingly, Theorem \ref{thm:Uhl} states that a density operator $\G'\in \E$ can be written as a convex combinations of $\G_i \in \Ew$ (actually $\G_i \equiv U_i \G U_i^\dagger$) if and only if its spectrum is majorized by $\wb$. Relation \eqref{ebwmajor} follows then directly from the linearity of the partial trace and Eq.~\eqref{Ebwmajor}. To prove \eqref{Fbwmajor}, we use the fact that each element $\G \in \Ebw \equiv \mbox{conv}(\Ew)$ can be written as a convex combination $\sum_i p_i \G_i$ of $\G_i \in \Ew$, leading to
\begin{eqnarray}
\lefteqn{\min_{\EbwS \ni \G \mapsto \g} \mathrm{Tr}[\G W]} \hspace{1cm} \nonumber  \\
&=&\min_{\scriptsize\begin{array}{c}\sum_i p_i\G_i \mapsto \g, \\ \G_i \in \Ew \end{array}} \sum_i p_i \mathrm{Tr}[\G_i W] \nonumber  \\
&=& \min_{\scriptsize\begin{array}{c}\sum_i p_i\g_i=\g, \\ \g_i \in \ew \end{array}} \min_{\{\Ew \ni \G_i \mapsto \g_i\}}\sum_i p_i \mathrm{Tr}[\G_i W] \nonumber \\
&\equiv&  \min_{\scriptsize\begin{array}{c}\sum_i p_i\g_i=\g, \\ \g_i \in \ew \end{array}} \sum_i p_i \Fw(\g_i) = \Fbw(\g)\,.
\end{eqnarray}
\end{proof}

In the following, we refer to $\g \in \ew$ as being \emph{$\wb$-ensemble $N$-representable} and to $\g \in \ebw$ as being \emph{relaxed $\wb$-ensemble $N$-representable}.

\section{Characterization of $\ebw$}\label{sec:characterize}
We start with two general comments on functional theories. Clearly, without knowing the domain of the universal functional, the common process of constructing and testing functional approximations cannot be initiated. Moreover, a formal definition of the domain $\ebw$ as in Eqs.~\eqref{E1},\eqref{ebw} is neither sufficient. Instead, a description of $\ebw$ is required which allows one in practical minimizations algorithms to verify without much effort whether a given 1RDM $\g$ belongs to $\ebw$ or not.
Deriving such a compact description of $\ebw$ is exactly the goal of the following two sections.
Due to the duality correspondence explained in Sec.~\ref{sec:dual}, this goal can be achieved by studying for all Hermitian operators $h$ on the one-particle Hilbert space $\mathcal{H}_1$ the minimization of $\mbox{Tr}_1[h \g]$ on the convex compact set $\ebw$.
It is exactly this dual perspective on convex compact sets which establishes a fruitful equivalence between the problem of characterizing $\ebw$
and describing various systems of $N$ \emph{non-interacting} fermions with arbitrary one-particle Hamiltonians $h$.

We first reveal that $\ebw$ takes effectively the form of a convex polytope whose vertex representation we determine (recall Sec.~\ref{sec:polytopes}). Then, in Sec.~\ref{sec:vtoh} we turn this into a hyperplane representation, leading to a hierarchical generalization of Pauli's famous exclusion principle to ensemble states.

Before we begin to derive a compact description of $\ebw$ for arbitrary $\wb$, let us discuss the specific case of $\wb = \wb_0$.
It has been shown by Kuhn and Coleman \cite{K60,C63} that $\g \in \e\equiv \,\xbar{\mathcal{E}}^1_N(\wb_0)$ if and only if its eigenvalues (natural occupation numbers) fulfill the Pauli exclusion principle \eqref{PC}. By referring to the vector majorization, this can also expressed
as
\begin{equation}\label{PCmajor}
\mbox{spec}(\g) \prec (1,\ldots,1,0,\ldots)\,.
\end{equation}
In agreement with our previous comments, the constraints on $\ebw$ for $\wb = \wb_0$ are purely spectral and the natural orbitals of the 1RDM $\g$ have no influence on its $N$-representability.
Moreover, it can easily be shown that the extremal points of $\e$ are exactly those 1RDMs which emerge from single Slater determinants \cite{W39,K60,CY00,LS09}. In other words, the extremal 1RDMs are characterized by $\mbox{spec}(\gamma)^\downarrow=(1,\ldots,1,0,0,\ldots)$.
But what are the extremal points of $\ebw$ for an arbitrary $\wb \neq \wb_0$ and how could we calculate them?
In the following, we are answering those crucial questions in a constructive way.
\begin{figure}[htb]
\frame{\includegraphics[scale=0.49]{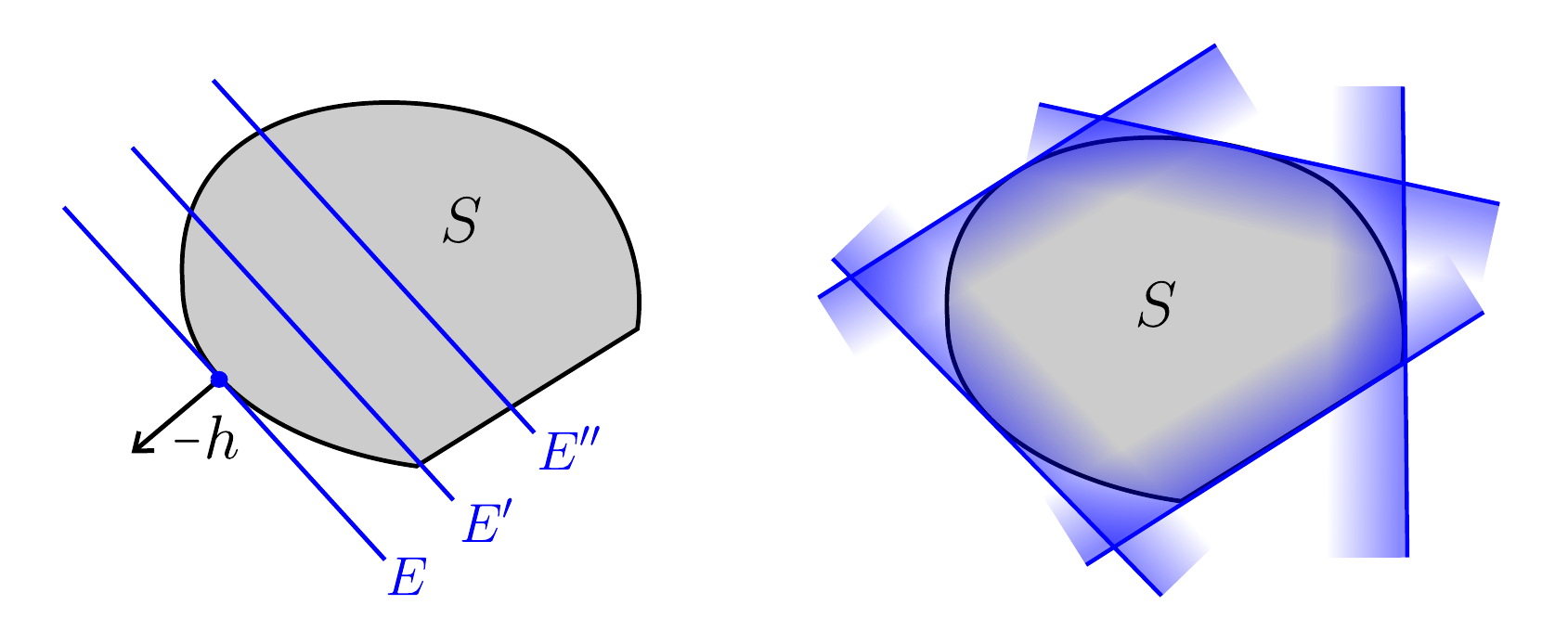}}
\caption{Geometric illustration of the duality correspondence, Theorem \ref{thm:dual}, for convex compact sets $S$. Left: the minimization of the linear function $\langle \cdot, h \rangle$ over $S$ means to shift the hyperplane (blue) of constant value $\langle \g, h\rangle$ along the direction of the normal vector $-h$ until the boundary is reached. Right: Realizing this minimization for all `directions' $-h$ fully characterizes the boundary of $S$ and therefore also $S$.}
\label{fig:dual}
\end{figure}

\subsection{Procedure for full characterization of $\ebw$}\label{sec:procedure}
The full characterization of $\ebw$ involves eight consecutive steps. This will reveal that $\ebw$ effectively takes the form of a convex polytope whose vertex representation can be derived in a systematic manner.

\subsubsection{Duality correspondence}
According to the duality correspondence presented in Sec.~\ref{sec:dual} (recall in particular Theorem \ref{thm:dual}) we can characterize the convex compact set $\ebw$ through the support function $\chi_{\ebwS}^\ast(h)$. Equivalently, it is sufficient to determine the minimizer $\g_h$ of $\langle h, \g\rangle_1 \equiv \mbox{Tr}_1[h \g]$ on $\ebw$ for any Hermitian one-particle operator $h$. We illustrate this minimization in Fig.~\ref{fig:dual}. The convex hull of all those minimizers (see right panel of Fig.~\ref{fig:dual}) coincides with $\ebw$. In a mathematically more formal way this can be stated as
\begin{equation}
\ebw = \mbox{conv}\Big(\!\Big\{\!\argmin_{\g \in \ebwS} \mbox{Tr}_1[\g h]\,\Big|\, h= h^\dagger\Big\}\!\Big).
\end{equation}
Actually, all those $\g \in \partial \ebw$ which emerge as unique minimizers of at least one $h$ are extremal points of $\ebw$.

\subsubsection{Lifting minimization to $N$-fermion level}
To access the information on the fixed spectral vector $\wb$ we lift this minimization process to the $N$-fermion level. This is based on the following (trivial) relation
\begin{equation}\label{Min1liftN}
\min_{\g \in \ebwS} \mbox{Tr}_1[h \g] = \min_{\G \in \EbwS} \mbox{Tr}_N[h \G]\,,
\end{equation}
where on the right-hand side $h$ should denote the one-particle operator from the left-hand side but lifted to the $N$-particle level (for the sake of simplicity we used the same mathematical symbol).

\subsubsection{Application of GOK variational principle}
Following the outline of Sec.~\ref{sec:relax} we have
\begin{equation}\label{ENrelx}
\min_{\G \in \EbwS} \!\!\mbox{Tr}_N[h \G] = \min_{\G \in \Ew}\!\! \mbox{Tr}_N[h \G]\,.
\end{equation}
The GOK variational principle, Theorem \ref{thm:GOK}, applied to the right-hand side of Eq.~\eqref{ENrelx} implies in combination with Eqs.~\eqref{Min1liftN},\eqref{ENrelx}
\begin{equation}\label{GOK0int}
\min_{\g \in \ebwS}\!\!\mbox{Tr}_1[h \g] =\sum_{i} w_i \tilde{E}_i\,.
\end{equation}
To avoid any confusion with the eigenvalues $E_i(h)$ of the interacting Hamiltonian \eqref{ham} we denote the eigenenergies of the non-interacting Hamiltonian $h$ by $\tilde{E}_i$.
Equation \eqref{GOK0int} means that the crucial expression $\min_{\g \in \ebwS} \mbox{Tr}_1[h \g]$ follows as the $\wb$-weighted sum of the eigenenergies $\tilde{E}_i$ of the  Hamiltonian $h$ on the $N$-fermion Hilbert space $\mathcal{H}_N$.  Even more importantly, the possible minimizers of the left-hand side in Eq.~\eqref{Min1liftN} follow as the 1RDMs of the $\wb$-weighted mixture of the eigenstates (ordered according to their energies) of $h$.

The interpretation of the mathematical problem on the right side of \eqref{ENrelx} as a minimization problem for $N$ non-interacting fermions with one-particle Hamiltonian $h$ is a key insight. This fruitful analogy will namely allow us to solve that mathematical problem in a rather elegant way by referring to our physical understanding of non-interacting fermion systems.

\subsubsection{Configuration states}
Since $h$ is a one-particle operator its eigenstates $\ket{\Psi_j}$ on the $N$-fermion Hilbert space are given by configuration states/Slater determinants,
\begin{equation}\label{SD}
\ket{\bd{i}}\equiv \ket{i_1,\ldots, i_N} \equiv f_{i_1}^\dagger \ldots f_{i_N}^\dagger \ket{0}\,.
\end{equation}
Here, by exploiting the formalism of second quantization, $\ket{0}$ denotes the vacuum state and $f_j^\dagger$ the fermionic creation operators referring to the eigenbasis $\{\ket{i}\}_{i=1}^d$ of $h$ as an operator on the one-particle Hilbert space,
\begin{equation}\label{h}
h\big|_{\Ho} = \sum_{i=1}^{d}h_i\, \big(f_i^\dagger f_i\big)\big|_{\Ho}\equiv \sum_{i=1}^{d}h_i \ket{i}\!\bra{i}\,.
\end{equation}
Moreover, we introduce the set
\begin{equation}\label{config}
\mathcal{I}_{N,d}\equiv \{\bd{i}\equiv (i_1,\ldots,i_N)\,|\, 1\leq i_1<i_2<\ldots <i_N \leq d\}\,,
\end{equation}
of all $N$-fermion configurations $\bd{i}$.

Clearly, by restricting to operators $h$ with unique minimizers $\g_h$, the eigenbases of $h$ and $\g_h$ necessarily coincide.

\subsubsection{Introduction of spectral polytopes}
The set $\ebw$ is invariant under unitary conjugation, i.e.,
\begin{equation}
u\, \ebw\, u^\dagger = \ebw
\end{equation}
for all unitary operators $u$ on the one-particle Hilbert space. This allows us to restrict to one fixed choice of (ordered) natural orbitals.
Based on this relation, we define for the following
\begin{eqnarray}\label{Sigma}
\Sigma^{\downarrow}(\wb) &\equiv& \mbox{spec}^\downarrow\big(\ebw\big)  \nonumber \\
\Sigma(\wb) &\equiv& \mbox{spec}\big(\ebw\big) \,.
\end{eqnarray}
Here, $\mbox{spec}\big(\ebw\big)$ shall denote the set of all possible spectra of 1RDMs $\g \in \ebw$ whose entries (natural occupation numbers) are arranged in any possible order.
By introducing the Pauli simplex
\begin{equation}\label{PCsimplex}
\Delta= \{\bd{\lambda}\in \RR^d\,|\,1\geq \lambda_1\geq \lambda_2\geq \ldots \geq \lambda_d\geq 0\}
\end{equation}
we could express the relation between both spectral polytopes \eqref{Sigma} in the following geometric way
\begin{equation}\label{Sigmarel}
\Sigma^{\downarrow}(\wb) = \Sigma(\wb)\cap \Delta\,.
\end{equation}
We illustrate $\Sigma(\wb)$ and $\Sigma^\downarrow(\wb)$ in Fig.~\ref{fig:ew0}: For the setting $(N,d)=(2,3)$ we present the two-dimensional spectral polytopes of the largest two natural occupation numbers $(\lambda^\downarrow_1,\lambda^\downarrow_2)$, where $\lambda^\downarrow_3$ follows from the normalization condition, $\lambda^\downarrow_1+\lambda^\downarrow_2+\lambda^\downarrow_3=2$.
Minimization of $h_1\lambda_1^\downarrow+h_2\lambda_2^\downarrow+h_3\lambda_3^\downarrow=(h_1- h_3)\lambda_1^\downarrow+(h_2- h_3)\lambda_2^\downarrow + 2h_3$ for all $\bd{h}$ with $h_1\leq h_2 \leq h_3$ (range shown in red) always leads to a single generating vertex $\bd v$. For the chosen weight vector $\bd w = (0.5, 0.4, 0.1)$ in Fig.~\ref{fig:ew0} this yields $\bd{v}=(0.9,0.6,0.5)$. Permutation of the entries of $\bd{v}$ leads to the remaining five vertices of the permutohedron $\mathcal{P}_{\bd{v}}=\Sigma(\wb)$. Anticipating the results of Sec.~\ref{sec:examples}, the polytope $\Sigma^\downarrow(\wb)$ is described by the ordering constraints $\lambda^\downarrow_1 \geq \lambda^\downarrow_2 \geq \lambda^\downarrow_3=2-\lambda^\downarrow_1-\lambda^\downarrow_2$ complemented by the two generalized exclusion principle constraints $\lambda^\downarrow_1+\lambda^\downarrow_2 \leq 1+w_1$ and $\lambda^\downarrow_1 \leq w_1+w_2$.

\begin{figure}[htb]
\includegraphics[scale=0.55]{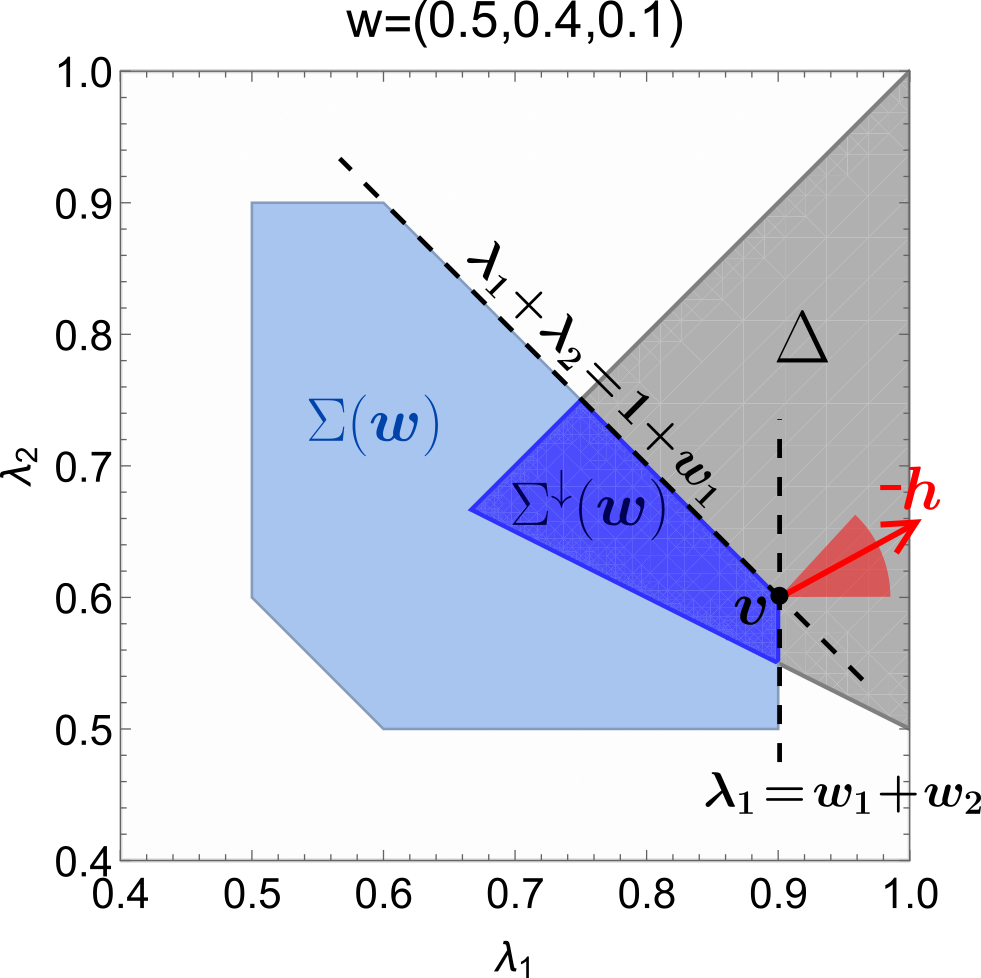}
\caption{Illustration of spectral polytopes $\Sigma(\wb)$ (light blue) and $\Sigma^{\downarrow}(\wb)= \Sigma(\wb) \cap \Delta$ (blue) \eqref{Sigma} for an exemplary $\wb$ in the setting $(N,d)=(2,3)$. Pauli simplex $\Delta$ \eqref{PCsimplex} shown in gray and normalization is used to substitute $\lambda_3=2-\lambda_1-\lambda_2$.}
\label{fig:ew0}
\end{figure}

\subsubsection{Generating vertices}
To calculate $\Sigma^{\downarrow}(\wb)$ and $\Sigma(\wb)$, respectively, we just need to determine the vector $\bd{\lambda}^\downarrow$ of decreasingly ordered natural occupation numbers for the minimizer states $\G$ for all possible Hamiltonians $h$. Without loss  of generality the class of those Hamiltonians \eqref{h} can be restricted to those
with an arbitrary but \emph{from now onwards fixed} eigenbasis $\mathcal{B}_1\equiv \{\ket{i}\}_{i=1}^d$ and $h_1 \leq h_2 \leq \ldots \leq h_d$. Independent of the number $r$ of non-vanishing weights, this consideration will always lead to a finite number of corresponding vectors $\bd{v}^{(j)}$, $j=1,2,\ldots, R<\infty$ of natural occupation numbers. In retrospective, this will prove our claim that the spectral sets $\Sigma(\wb)$ and $\Sigma^{\downarrow}(\wb)$ take the form of convex polytopes. Due to the ordering $h_i \leq h_{i+1}$ and $w_i \geq w_{i+1}$, the entries of any vector $\bd{v}^{(j)}$ will be automatically ordered decreasingly. The spectral polytope $\Sigma(\wb)$ then follows as
\begin{equation}\label{polyvertices}
\Sigma(\wb) = \mbox{conv}\big(\big\{\pi(\bd{v}^{(j)})\,\big|\, j=1,\ldots,R,\, \pi \in \mathcal{S}^d\big\}\big)\,.
\end{equation}
Since $\Sigma(\wb)$ is a permutation-invariant polytope we can restrict ourselves to $\Sigma^\downarrow(\wb)$. Expressing $\Sigma^\downarrow(\wb)$ in the hyperplane representation will then lead to a rather compact description of the spectral polytope, which after all will turn out to be effectively independent of $N,d$. We present several examples illustrating the hyperplane representation of $\Sigma^\downarrow(\wb)$ in Sec.~\ref{sec:examples}.

\subsubsection{Partial ordering of configurations}
To determine all extremal points of $\Sigma(\wb)$, or to be more precise all vectors $\bd{v}^{(j)}$, we just need to determine all possible $\wb$-ensemble minimizer states $\G$ for the specific class of Hamiltonians $h$ as defined in point 6. Due to the separation of the natural occupation numbers and the irrelevant natural orbitals this eventually amounts to a purely combinatorial problem. For practically relevant (i.e., not too large) $r$ this can be solved in a straightforward manner. To explain this, we introduce for each set of energy levels
\begin{equation}\label{hvec}
\bd{h}\equiv \bd{h}^{\uparrow}\equiv (h_1,\ldots, h_d)
\end{equation}
a \emph{linear  ordering} $\leq_{\bd{h}}$ on the set $\mathcal{I}_{N,d}$ \eqref{config} of $N$-fermion configurations by
\begin{equation}\label{linearorder}
\bd{i} \leq_{\bd{h}} \bd{j} \quad :\Leftrightarrow \quad \tilde{E}_{\bd{i}} \leq \tilde{E}_{\bd{j}}\,,
\end{equation}
where $\tilde{E}_{\bd{i}}\equiv \bra{\bd{i}}h\ket{\bd{i}}=\sum_{i \in \bd{i}} h_i$.
For any fixed $\bd{h}$, this leads to a linear hierarchy of all configurations in  $\mathcal{I}_{N,d}$. Actually, independent
of the values $h_1 \leq h_2 \leq \ldots \leq h_d$ the first (lowest) one is always $(1,2,\ldots, N)$, which is always followed by
$(1,\ldots,N-1,N+1)$. The third lowest (and the following ones) depend on the values $h_i$ though.
We therefore would need to determine for some fixed $r$, by considering various $\bd{h}$, all possible sequences of the lowest $r$ configurations,
\begin{equation}\label{seqh}
\bd{i}_1\leq_{\bd{h}} \bd{i}_2 \leq_{\bd{h}}\ldots \leq_{\bd{h}} \bd{i}_r \,,
\end{equation}
where $\bd{i}_1=(1,\ldots,N)$ and $\bd{i}_2=(1,\ldots,N-1,N+1)$. There will be only finitely many different sequences. Each of them leads to one vector $\bd{v}$ according to
\begin{equation}\label{vvsni}
\bd{v}= \sum_{j=1}^r w_j \bd{n}_{\bd{i}_j}\,,
\end{equation}
where we introduced the occupation number vector $\bd{n}_{\bd{i}}$ of the configuration $\bd{i}$. The $k$-th entry of that vector is one whenever $k$ is occupied, $k \in \bd{i}$, and zero otherwise. For instance, we have $\bd{n}_{(1,\ldots,N)}=(1,\ldots,1,0,\ldots)$.
The examples provided in Sec.~\ref{sec:examples} will illustrate this procedure of determining in a systematic way the vertex representation of the spectral polytope $\Sigma(\wb)$.

To determine for a fixed $r$ all those sequences, we recall that some pairs of configurations $\bd{i}, \bd{j}$ share a relation regardless of the values of $\bd{h}\equiv \bd{h}^\uparrow$. For instance, as already explained above, any $\bd{j}$ is related to the configuration $(1,\ldots,N)$, as $(1,\ldots,N) \leq_{\bd{h}} \bd{j}$, for all $\bd{h}$. Such universal relations simplify our task of identifying all sequences \eqref{seqh} of length $r$ considerably and we introduce therefore the partial ordering on $\mathcal{I}_{N,d}$
\begin{eqnarray}\label{order}
\bd{i} \leq \bd{j} \quad &:\Leftrightarrow& \quad  \bd{i} \leq_{\bd{h}} \bd{j}\,, \quad \forall \bd{h}\equiv \bd{h}^\uparrow  \nonumber \\
&\,\,\Leftrightarrow&\sum_{k=1}^N h_{i_k} \leq \sum_{k=1}^N h_{j_k}\,, \quad \forall \bd{h}\equiv \bd{h}^\uparrow \,.
\end{eqnarray}
Clearly, one has $\bd{i} \leq \bd{j}$ if and only if $i_k \leq j_k$ for all $k=1,2,\ldots,N$. Hence, Eq.~\eqref{order} provides the connection between the elements of the $N$-tuples $\bd i$ and the energies $h_1\leq h_2\leq \ldots\leq h_d$.

\begin{figure}[htb]
\frame{\includegraphics[scale=0.57]{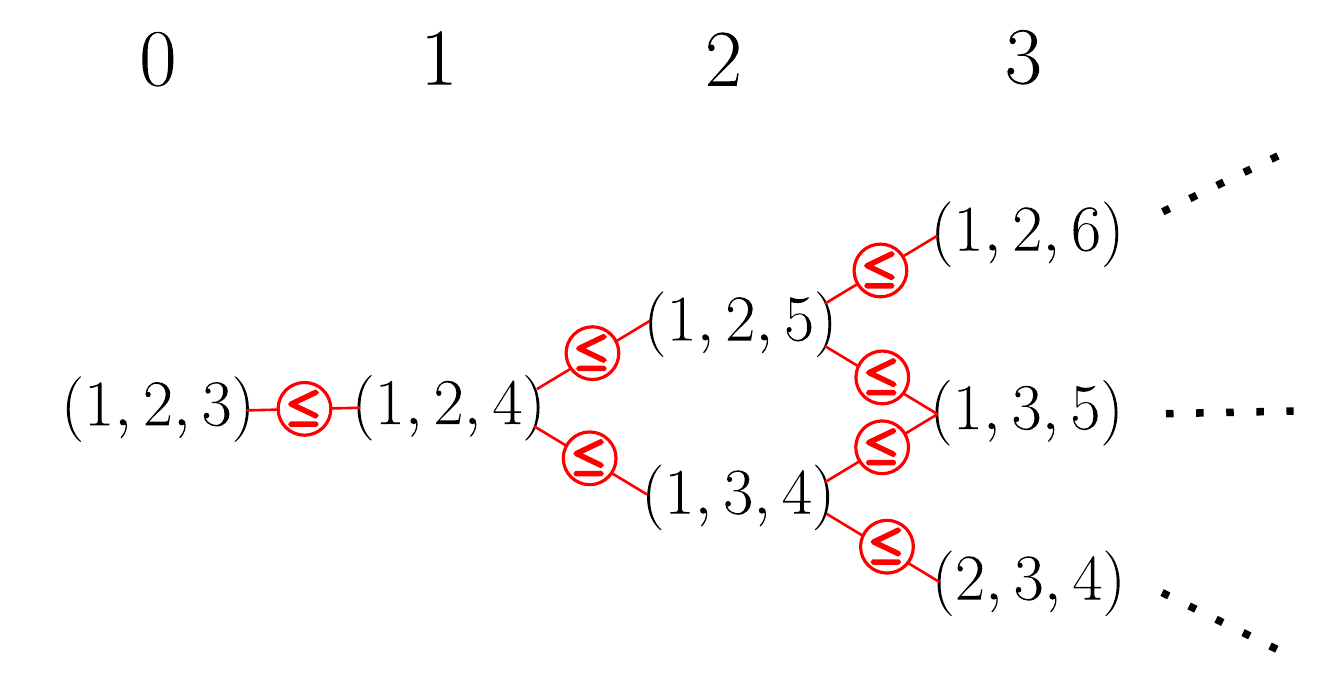}}
\caption{Illustration of the partial ordering \eqref{order} between the energetically lowest few configurations $\bd{i}$ of
the one-particle Hamiltonian $\bd{h} \equiv \bd{h}^\uparrow$ for the case of $N=3$ fermions (see text for more details).}
\label{fig:excit}
\end{figure}
The partial ordering $\leq$ of the lowest few configurations is illustrated in Fig.~\ref{fig:excit} for the case $N=3$. A huge advantage of our whole procedure is that such excitation patterns are effectively independent of $N,d$, as long as
\begin{equation}\label{Ndmin}
N \geq r-1\quad \mbox{and}\quad d-N \geq r-1\,.
\end{equation}
To be more specific, drawing for arbitrary $N$ and $d$ which obey \eqref{Ndmin} the analogous pattern of the one shown in Fig.~\ref{fig:excit} is trivial: One just needs to increase the orbital index in each configuration by $N-3$ and then to add to each configuration the orbital indices $1,2,\ldots,N-3$. In the opposite case, when $N$ or $d-N$ is reduced (for fixed $r$) to violate condition \eqref{Ndmin}, it just means that some of the sequences \eqref{seqh} are left out. For these cases $N,d$ one would need to redo parts of the derivation. Yet, this would be a rather elementary exercise since $N$ and $d-N$ are small. Moreover, this effective independence of $N,d$ of the excitation pattern turns out to be the reason why the hyperplane representation derived in Sec.~\ref{sec:vtoh} will be effectively independent of $N$ and $d$ as well.

\subsubsection{Putting everything together}
Last but not least, we work out the sequences
\begin{equation}\label{sequence}
\bd h\mapsto \Gamma \mapsto \gamma\mapsto \bd v\,,
\end{equation}
to obtain all the required vertices for generating the spectral polytope $\Sigma(\wb)$ according to Eq.~\eqref{polyvertices}. For this, we recall that in virtue of Eq.~\eqref{seqh} every vector $\bd h \equiv \bd{h}^\uparrow $ determines a lineup
\begin{equation}
\bd i_1\to \bd i_2 \to \ldots\to \bd i_r\,,
\end{equation}
given by the first $r$ configurations $\bd  i_1,\ldots, \bd  i_r$ with respect to the linear order $\leq_{\bd{h}}$.
Consequently, for any fixed $r$ there will be only finitely many such lineups and thus only finitely many generating vertices $\bd{v}^{(j)}$.

For every lineup, the corresponding $N$-fermion density operator $\Gamma$ follows as
\begin{equation}
\Gamma = \sum_{j=1}^r w_j\ket{\bd i_j}\!\bra{\bd i_j}\,,
\end{equation}
and tracing out $N-1$ particles yields the corresponding $\gamma$. Since the occupation number vectors $\bd{n}_{\bd{i}}$ occurring in Eq.~\eqref{vvsni} are nothing else than the spectra of the 1RDMs of the configuration states $\ket{\bd{i}}$, we eventually obtain
\begin{equation}
\bd  v = \sum_{j=1}^r w_j\mathrm{spec}\big(N\Tr_{N-1}\left[\ket{\bd i_j}\!\bra{\bd i_j}\right]\big)\,.
\end{equation}

From a mathematical point of view, we would be done now. We namely succeeded in determining the spectral polytope $\Sigma(\wb)$ \eqref{Sigma} according to Eq.~\eqref{polyvertices} and therefore also $\ebw$. Yet, from a practical point of view, this $v$-representation of $\Sigma(\wb)$ is of little use. To be more specific, when minimizing an approximate functional by a gradient method one would need to know when the vector $\bd{\lambda}$ is leaving the polytope $\Sigma(\wb)$. In striking contrast to the $v$-representation this can easily be verified within the $h$-representation by checking whether all linear hyperplane conditions $D_k(\bd{\lambda})\geq 0$ are still fulfilled. Another huge advantage of the $h$-representation is that the number of inequalities (and even their form) turns out to be independent of $N$ and $d$. This is in striking contrast to the number of vertices of $\Sigma(\wb)$ which grows roughly as $d^{N+r-1}/(N-r+1)!$ and therefore would diverge in the complete basis set limit $d\rightarrow \infty$.

\subsection{Translating $\Sigma(\wb)$'s $v$-representation into an $h$-representation}\label{sec:vtoh}
As far as the practical application of $\wb$-ensemble RDMFT is concerned, we still need to translate the vertex representation of the permutation-invariant polytope $\Sigma(\wb)$ into the (expected) compact hyperplane representation. This means, starting from the definition of the generating vertices in Eq.~\eqref{vvsni}, we would like to derive the minimal set of facet-defining inequalities. Restricting in this context to vectors $\bd{\lambda}^\downarrow$ of decreasingly ordered natural occupation numbers $\lambda_i^\downarrow$ will drastically reduce this set of inequalities. In the following, we explain the strategy for this and discuss general properties of the hyperplane representation, while the inequalities for different numbers $r$ of finite weights $w_1,\ldots, w_r$ are presented and illustrated in the subsequent section, Sec.~\ref{sec:examples}.

We first consider the two most significant cases $r=1,2$. They correspond to the well known ground state RDMFT and an RDMFT for describing ground state gaps, respectively. As it has been explained in the previous section, there is only one lineup possible for $r=1,2$ and therefore only one generating vertex $\bd{v}$ \eqref{vvsni}. The spectral polytope $\Sigma(\wb)$ is therefore given as the permutohedron $\mathcal{P}_{\bd{v}}$ (recall
Eq.~\eqref{polyvertices}). According to Rado's theorem (Theorem \ref{thm:Rado}), the hyperplane representation follows immediately as
\begin{equation}\label{SigmaRado}
\Sigma(\wb) = \mathcal{P}_{\bd{v}}=\{\bd{\lambda}\in \RR^d\,|\,\bd{\lambda} \prec \bd{v}\}\,.
\end{equation}
The majorization condition $\bd{\lambda} \prec \bd{v}$ contains $d$ inequalities for $\bd{\lambda}^\downarrow$. This represents a drastic simplification of the description of the polytope compared to its $v$-representation with its huge number of vertices. As a matter of fact, most of those inequalities are even redundant (see also Sec.~\ref{sec:examples}).

For each $r>2$, there is more than one lineup of length $r$ (see also Fig.~\ref{fig:excit}) and thus $R>1$ many generating vectors $\bd{v}^{(j)}$. While $\Sigma(\wb)$ is therefore not a permutohedron anymore it is still permutation-invariant. To determine the hyperplane representation for larger $r$, we propose a crucial generalization of Rado's theorem
\begin{thm}\label{thm:Radogener}
Given finitely many vectors $\bd{v}^{(1)},\ldots, \bd{v}^{(R)}\in \RR^d$. The polytope
\begin{equation}
\mathcal{P} = \mathrm{conv}\big(\big\{\pi(\bd{v}^{(j)})\,\big|\, j=1,\ldots,R,\, \pi \in \mathcal{S}^d\big\}\big)
\end{equation}
can be characterized equivalently as
\begin{equation}
\mathcal{P} = \Big\{\bd{\lambda}\,\Big|\,\exists \,\mathrm{conv.~comb.}\,\sum_{i=1}^R p_i\bd{v}^{(j)}\equiv \bd{u}\!: \bd{\lambda} \prec \bd{u} \Big\}\,.
\end{equation}
\end{thm}
We present the proof of Theorem \ref{thm:Radogener} in \mbox{Appendix \ref{app:Rado2proof}}. Here, in the main text, we illustrate and motivate instead Theorem \ref{thm:Radogener} in the form of Fig.~\ref{fig:E1Nconst}. There, we present (schematically) a part of the polytope $\Sigma(\wb)$ with $R=2$ generating vertices $\bd{v}^{(1)}, \bd{v}^{(2)}$. Each of them generates a separate permutohedron $\mathcal{P}_{\bd{v}^{(i)}}$, shown in `blue' and `red' in the right panel. Apparently, taking the convex hull of the union of those two permutohedra leads to the polytope $\Sigma(\wb)$, $\Sigma(\wb) = \mbox{conv}(\mathcal{P}_{\bd{v}^{(1)}}\cup \mathcal{P}_{\bd{v}^{(2)}})$. This is nothing else than a reformulation of Theorem \ref{thm:Radogener}, where the convex combination of the vertices $\bd{v}^{(i)}$ defines a point on the line between $\bd{v}^{(1)}$ and $\bd{v}^{(2)}$.
\begin{figure}[htb]
\includegraphics[scale=0.8]{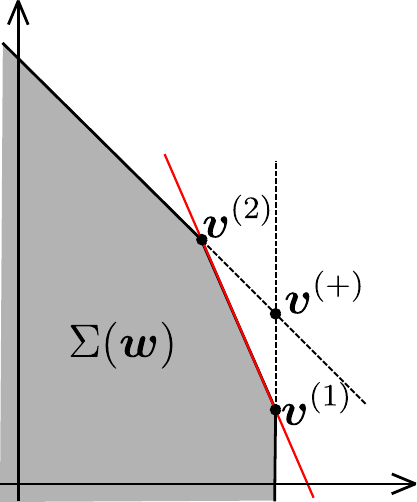}
\hspace{0.9cm}
\includegraphics[scale=0.8]{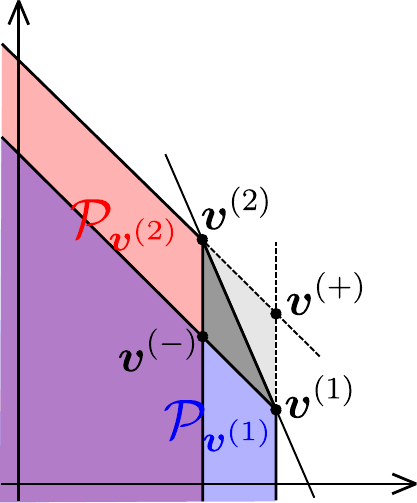}
\caption{Schematic illustration of $\Sigma(\wb)$ for $R=2$ generating vertices $\bd{v}^{(1)},\bd{v}^{(2)}$ and its inner and outer approximations \eqref{SigmaApprox} given by the
permutohedra $\mathcal{P}_{\bd{v}^{(-)}}$ and $\mathcal{P}_{\bd{v}^{(+)}}$, respectively.
The spectral polytope $\Sigma(\wb)$ is given by those $\bd{\lambda}$ which are majorized by some point on the line between the two vertices $\bd{v}^{(1)}, \bd{v}^{(2)}$ and $\mathcal{P}_{\bd{v}^{(-)}}$ (purple) follows as the intersection of the permutohedra $\mathcal{P}_{\bd{v}^{(1)}}$ (blue) and $\mathcal{P}_{\bd{v}^{(2)}}$ (red).
}
\label{fig:E1Nconst}
\end{figure}

In addition, Fig.~\ref{fig:E1Nconst} illustrates how a simple inner and outer approximation to the permutation-invariant polytope
$\Sigma(\wb)$ could be found, for an arbitrary number $R$ of generating vertices $\bd{v}^{(i)}$:
For a set of vectors $\bd{v}^{(1)},\ldots,\bd{v}^{(R)}$ whose entries are decreasingly ordered we consider the family of vectors which majorize all $\bd{v}^{(i)}$. Among that set there exists a vector $\bd{v}^{(+)}$ smallest with respect to $\prec$. Indeed, such a vector exists and it is  uniquely characterized by
\begin{equation}
\sum_{i=1}^k v_i^{(+)}  \equiv \max_{j}\left( \sum_{i=1}^k v_i^{(j)}\right)\,,\quad \forall k\,.
\end{equation}
In the same way, we introduce a maximal vector $\bd{v}^{(-)}$ within the set of those vectors which are majorized simultaneously by all $\bd{v}^{(j)}$,
\begin{equation}
\sum_{i=1}^k v_i^{(-)}  \equiv \min_{j}\left( \sum_{i=1}^k v_i^{(j)}\right)\,.
\end{equation}
One can then use the majorization conditions $\bd\lambda\prec \bd{v}^{(+)}$ and $\bd\lambda\prec \bd{v}^{(-)}$ to obtain an inner and outer approximation to the permutation invariant polytope $\Sigma(\wb)$,
\begin{equation}\label{SigmaApprox}
\mathcal{P}_{\bd{v}^{(-)}} \subset \Sigma(\wb) \subset \mathcal{P}_{\bd{v}^{(+)}}\,.
\end{equation}
Without too much effort, one can actually determine the vector $\bd{v}^{(+)}$ for arbitrary $N,d,r$. It follows as
\begin{equation}
\bd{v}^{(+)} = (\underbrace{1,\ldots,1}_{N-1},w_1,1-w_1,0,\ldots)\,.
\end{equation}
The vector $\bd{v}^{(-)}$, however, depends on the value $r$ and providing a concrete expression would be more tedious (it would in particular involve several floor functions). From a geometric point of view, $\mathcal{P}_{\bd{v}^{(-)}}$ is nothing else than the intersection of various permutohedra $\mathcal{P}_{\bd{v}^{(i)}}$,
\begin{equation}
\mathcal{P}_{\bd{v}^{(-)}}= \bigcap_{i=1}^R \mathcal{P}_{\bd{v}^{(i)}}\,.
\end{equation}

Before deriving the hyperplane representations of $\Sigma(\wb)$ in Sec.~\ref{sec:examples} for different settings $r$, we would like to discuss a remarkable connection between them. As an immediate consequence of Eq.~\eqref{ebwmajor}, an inclusion relation exists between the domains of related weight vectors $\wb, \wb'$,
\begin{equation}\label{inclusion}
\wb^\prime\prec\wb\quad\Leftrightarrow\quad\xbar{\mathcal{E}}^1_N(\wb^\prime)\subset\ebw\,.
\end{equation}
In Fig.~\ref{fig:ew}, we illustrate this inclusion relation for the exemplary case $N=2$ and $d=3$. For this, we present the spectral polytopes $\Sigma(\wb)$ and $\Sigma^\downarrow(\wb)$ for five weight vectors $\wb$ which decrease with respect to $\prec$ from left and right. Indeed, we observe that  the corresponding polytopes obey an inclusion relation in agreement with \eqref{inclusion}.

A more comprehensive mathematical analysis presented in our work \onlinecite{CLLS21} reveals an even stronger connection between the spectral polytopes $\Sigma(\wb)$ for different weight vectors $\wb$. Comparing their hyperplane representations for different values $r$ reveals a hierarchical generalization of Pauli's exclusion principle. For $r=1$, $\Sigma(\wb)$ is described by the Pauli exclusion principle constraint $\lambda_1^\downarrow \leq 1$. Then, by increasing the value $r$ step by step, more and more additional linear constraints on $\bd{\lambda}^\downarrow$ occur. This means that the inequalities for the setting $r'=r+1$ are given by those of the setting $r$, complemented by additional new ones. In particular, each generalized exclusion principle constraint derived for some setting $r$ will hold also for all larger settings $r'>r$ (and they remain facet-defining). This is remarkable since the constraints for the setting $r$ depend on a vector $\wb$ with $r$ non-vanishing weights $w_j$ summing up to one. Yet, these inequalities still hold for larger settings with some $\wb'$, by replacing $(w_j)_{j=1}^r$ by  $(w'_j)_{j=1}^r$, where the sum $\sum_{j=1}^r w_j'$ is typically smaller than one.
\onecolumngrid

\begin{figure}[htb]
\fbox{\includegraphics[scale=0.27]{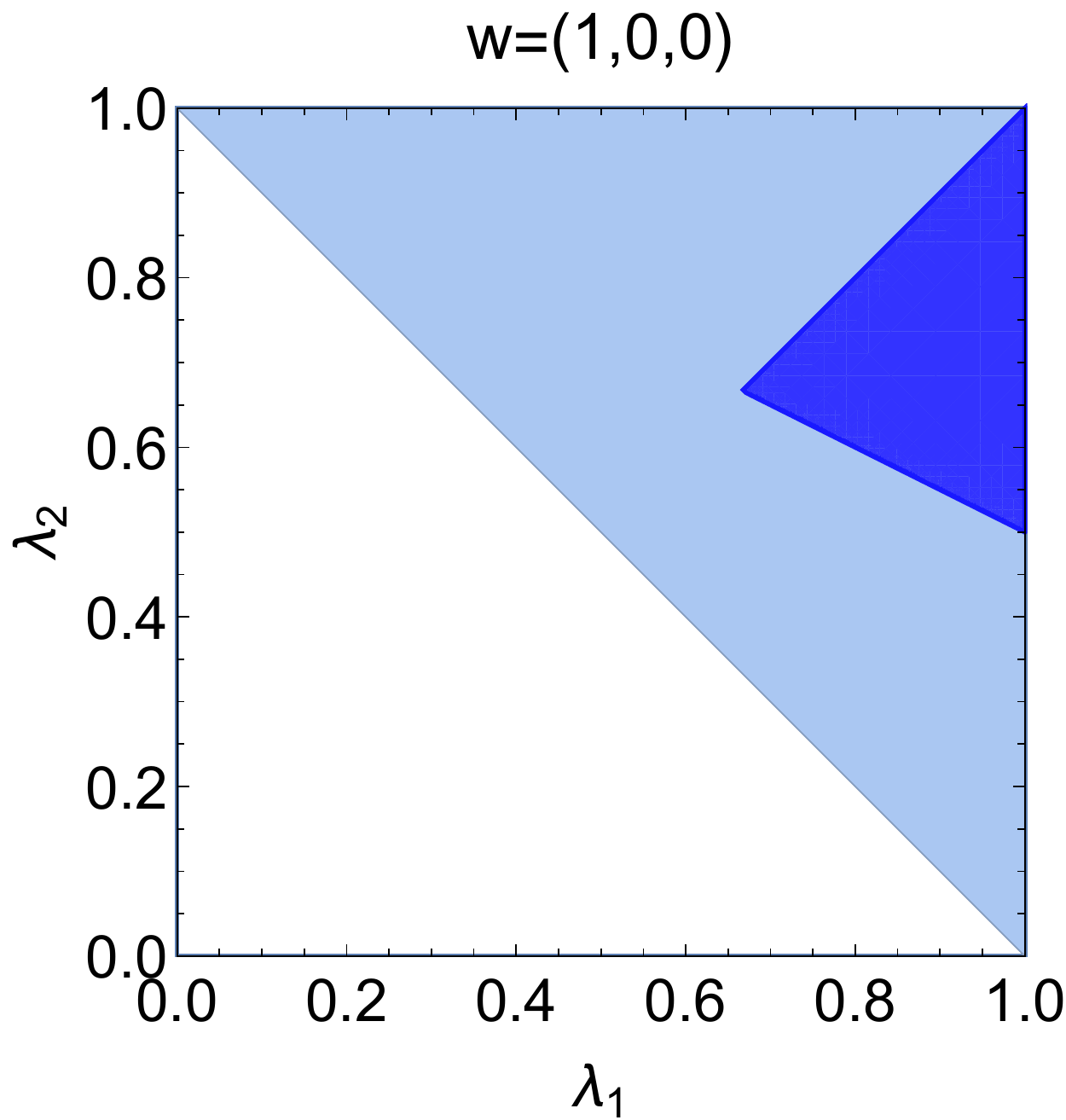}
\includegraphics[scale=0.27]{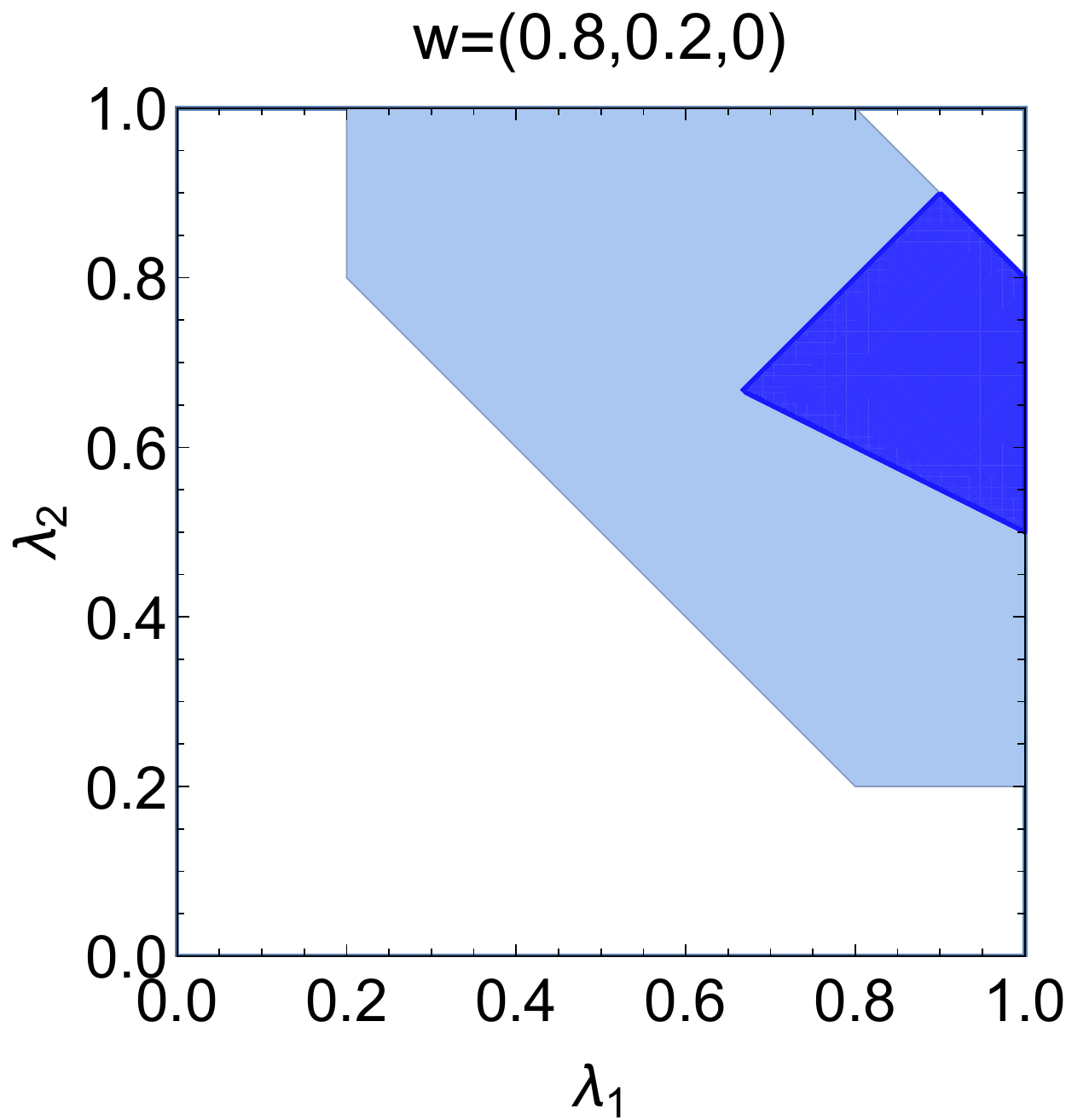}
\includegraphics[scale=0.27]{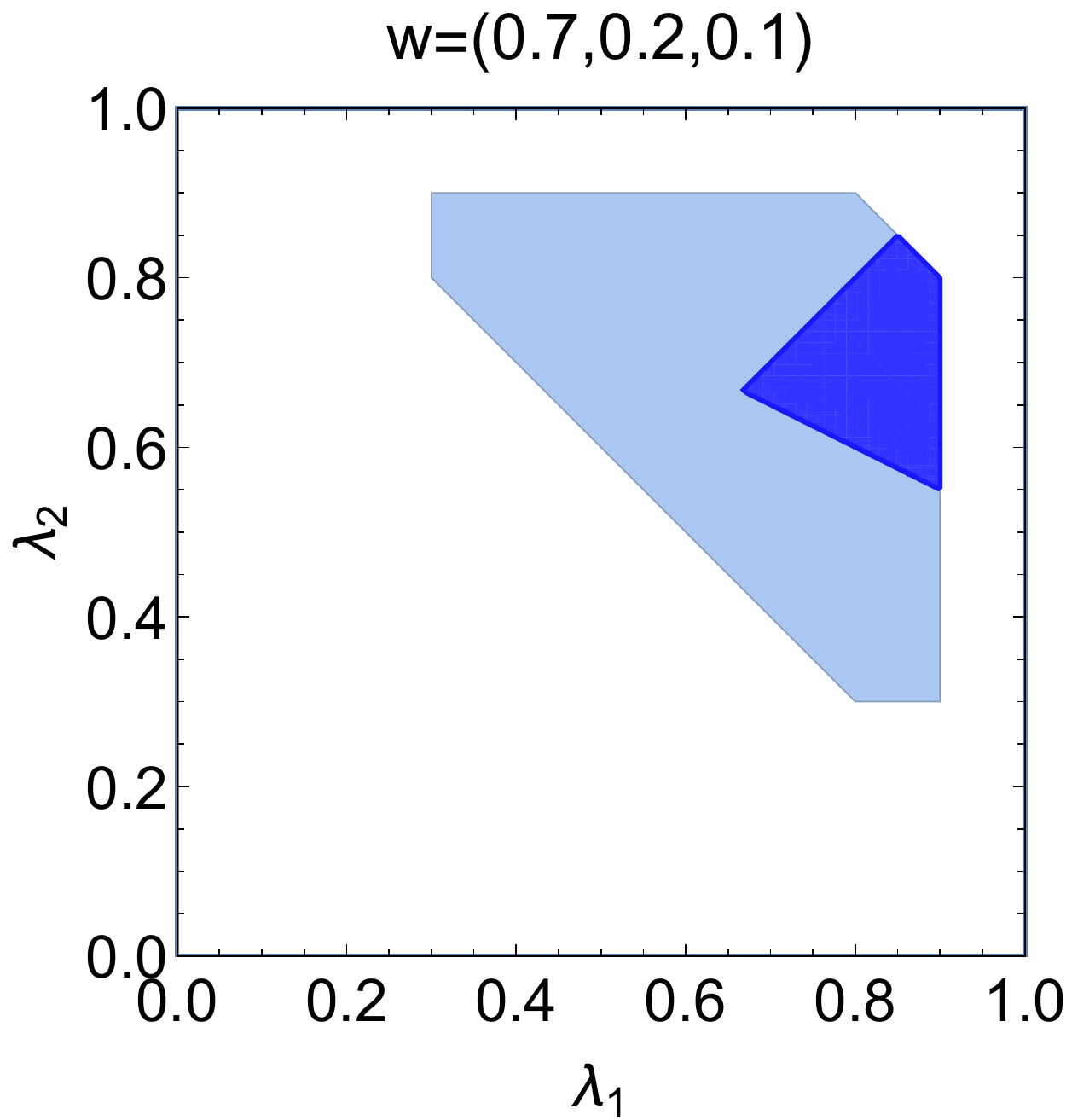}
\includegraphics[scale=0.27]{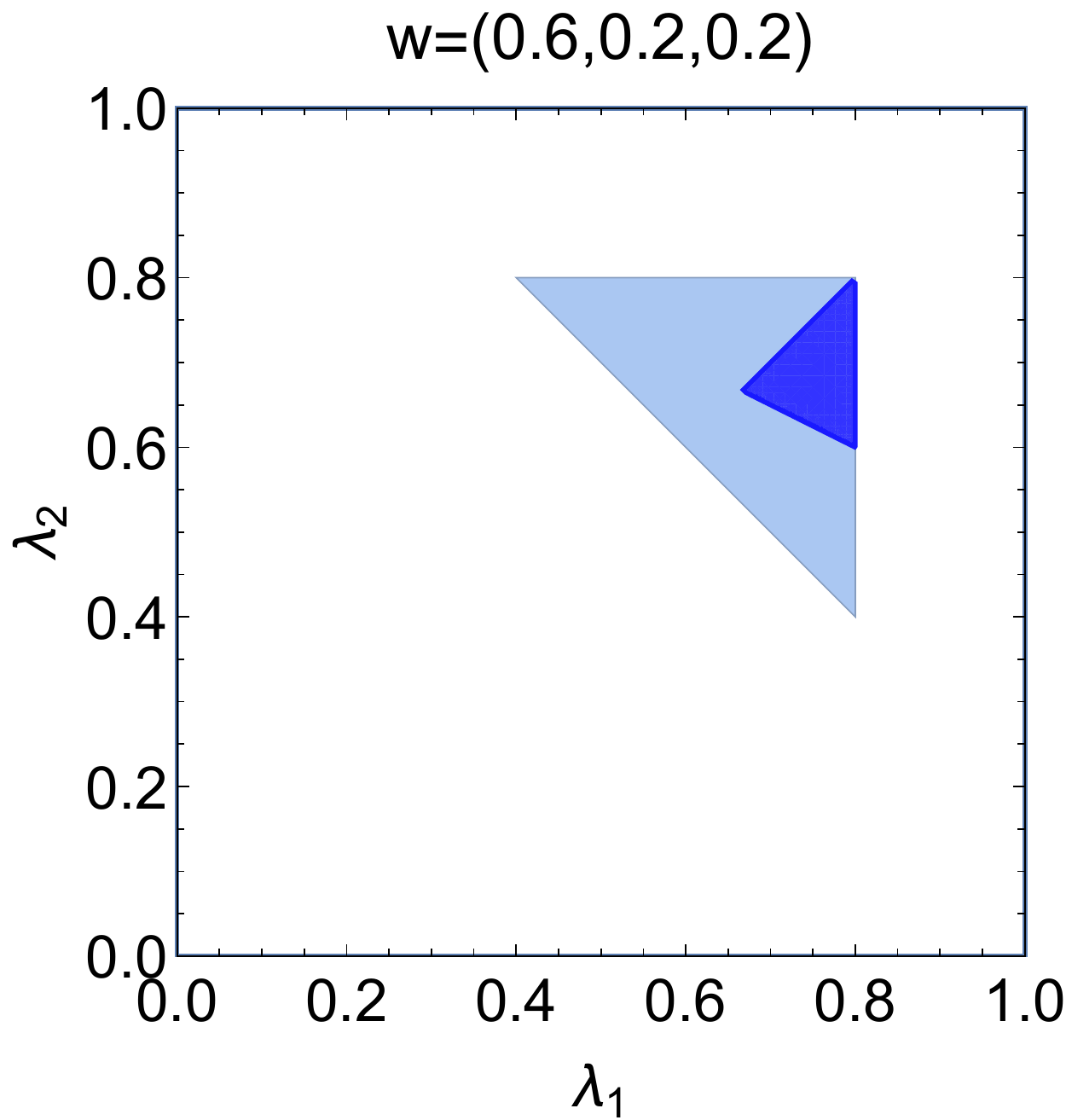}
\includegraphics[scale=0.27]{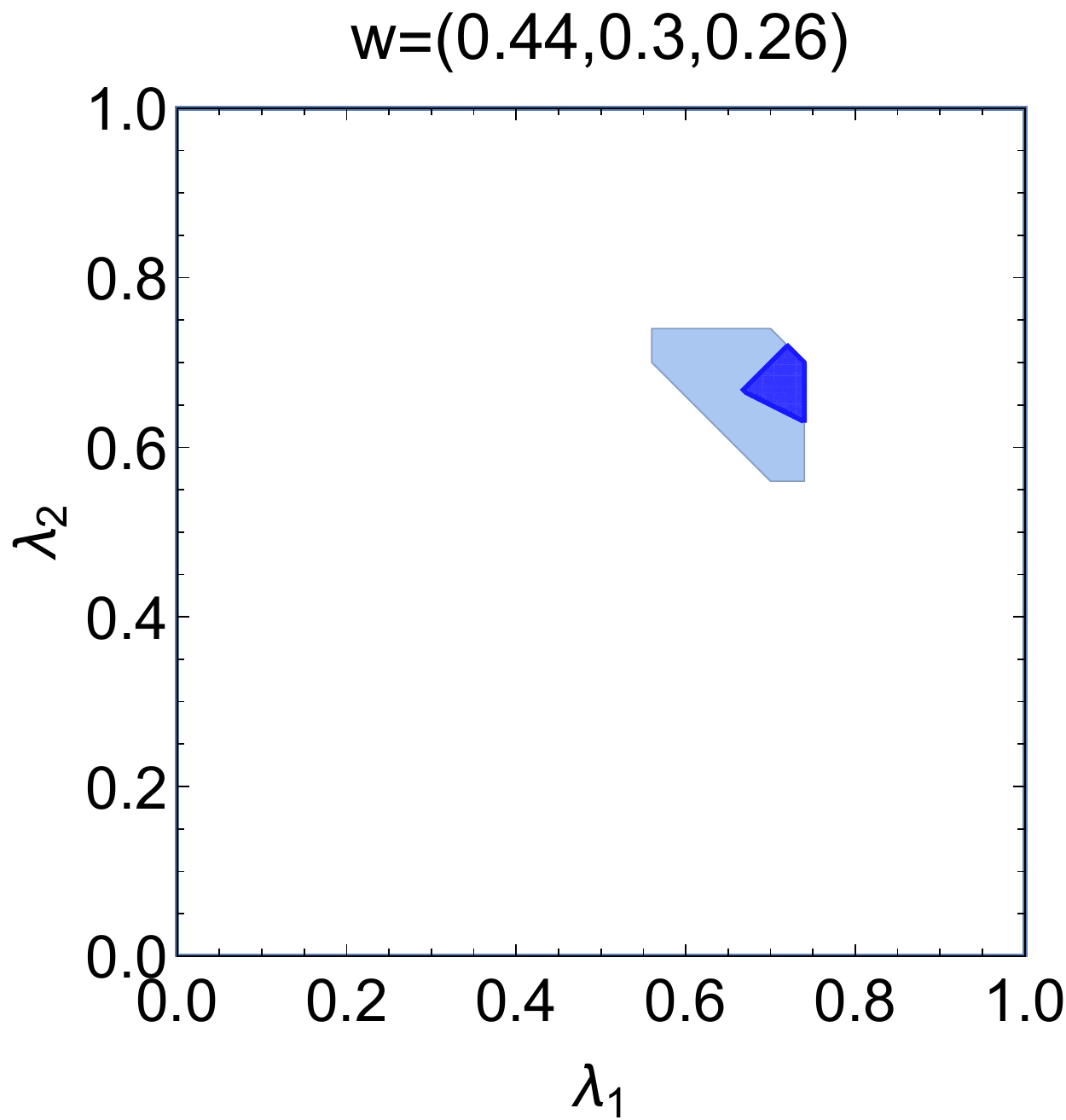}}
\caption{For the case of two fermions and a three-dimensional one-particle Hilbert space we present the spectral polytope $\Sigma(\wb)$ (light blue) \eqref{Sigma} of $\wb$-ensemble $N$-representable 1RDMs for five exemplary weight vectors $\wb$ to illustrate the inclusion relation in Eq.~\eqref{inclusion}. Normalization leads to a reduction of dimensionality according to  $\lambda_3=2-\lambda_1-\lambda_2$. By assuming a decreasing order of the natural occupation numbers, the permutation-invariant polytope $\Sigma(\wb)$ reduces to the polytope $\Sigma^{\downarrow}(\wb)$ (blue).}
\label{fig:ew}
\end{figure}
\twocolumngrid

A further crucial feature of the generalized exclusion principle is that increasing $N$ and $d$ does not add new inequalities, provided \eqref{Ndmin} is respected. Hence, $r$ determines the number of facet-defining inequalities entirely. This is a pleasant feature of the relaxed $\wb$-ensemble $N$-representability constraints compared to the generalized Pauli constraints \cite{KL06,AK08} whose complexity depends strongly on the setting $(N,d)$.
Hence, the hierarchy of generalized exclusion principle constraints is independent of the values of $N$ and $d$.
Actually, not only the number of inequalities is independent of $N$ and $d$ but even their form (see also the examples in the next section and in Ref.~\onlinecite{CLLS21}). We complete the present section by presenting the number of generating vectors $\bd v^{(i)}$ and inequalities for $r$ up to twelve in Table \ref{tab:r}.
\begin{table}[htb]
\resizebox{\linewidth}{!}{
\renewcommand{\arraystretch}{1.1}
\begin{tabular}{|c|c|c|c|c|c|c|c|c|c|c|c|c|}
\hline
r       & 1& 2& 3& 4&  5&  6&  7&   8& 9& 10& 11 &12  \\ \hline
$\#(\bd{v}^{(i)})$   & 1& 1& 2& 4& 10& 28& 90& 312& 1160&4518 & 18008 & 73224 \\ \hline
$\#(\text{ineq.})$ & 1& 2& 3& 5&  8& 13& 23&  42& 88 & 203& 486& 1257 \\ \hline
$\#(\text{new ineq.})$ & 1&1&1&2&3&5&10&19&46&115&283&771 \\ \hline
\end{tabular}
}
\caption{Number $\#(\bd{v}^{(i)})$ of generating vertices $\bd{v}^{(i)}$ and number $\#(\text{ineq.})$ of generalized exclusion principle constraints (relaxed $\wb$-ensemble N-representability conditions) on $\bd{\lambda}^\downarrow$ for the settings $r\leq 12$. A hierarchical generalization of Pauli's exclusion principle exists:  $\#(\text{new ineq.})$ denotes the number of new inequalities in the setting $r$ which did not exist yet in the setting $r-1$.}
\label{tab:r}
\end{table}

\section{Illustration and Examples}\label{sec:examples}
In this section, we derive and discuss the relaxed $\wb$-ensemble $N$-representability constraints for the cases \mbox{$r\leq 4$}. From a physical point of view, these are the most important cases since the corresponding $\wb$-ensemble RDMFT would allow one to determine the energies of the lowest four eigenstates. For each of those $r$, we first discuss the exemplary case of $N=3$ fermions and in that context refer to Fig.~\ref{fig:excit}.
Our general strategy (as explained in the previous section) will reveal that the restriction to $N=3$ is unnecessary and the results for arbitrary settings $(N,d)$ follow in a straightforward manner from those of $N=3$. Also, the generalized exclusion  principle constraints derived for any value $r$ will de facto be independent of the dimension $d$ and therefore the scope of our results will include the important complete basis set limit, i.e., $d\rightarrow \infty$.

\subsection{Relaxed $\wb$-ensemble $N$-representability conditions for $r\leq 2$ non-vanishing weights}\label{sec:r=2}
The simplest case, $r=1$ non-vanishing weights, corresponds to ground state RDMFT. As it is illustrated in Fig.~\ref{fig:excit} for $N=3$, there is a unique minimal configuration, namely $(1,2,3)$. The unique minimizer state is thus given by $\G=\ket{1,2,3}\!\bra{1,2,3}$ and the generating vertex follows as $\bd{v}=(1,1,1,0,\ldots)$ in agreement with Eq.~\eqref{PCmajor}. This majorization of a natural occupation number vector $\bd\lambda$ then takes the concrete form $1\geq \lambda_1^\downarrow\geq \lambda_2^\downarrow \geq  \ldots \geq \lambda_d^\downarrow\geq 0$, i.e., we recovered the well-known Pauli exclusion principle. 

For $r=2$ weights, there is only one lineup of length two possible, namely (see also Fig.~\ref{fig:excit})
\begin{equation}
\quad(1,2,3)\rightarrow (1,2,4)\,.
\end{equation}
This means that there is again only one minimizer $\G$ for the whole class of $\bd{h}$, with $h_1 \leq h_2 \leq \ldots < h_d$. It is given by (recall $w_1+w_2=1$)
\begin{equation}
\G= w_1\ket{1,2,3}\!\bra{1,2,3}+w_2\ket{1,2,4}\!\bra{1,2,4}
\end{equation}
and its vector of natural occupation numbers reads
\begin{equation}
\bd{v} = (1,1,w_1,w_2,0,\ldots)\,.
\end{equation}
The spectral polytope $\Sigma(\wb)$ is then given as the permutohedron $P_{\bd{v}}$. According to Rado's theorem, Theorem \ref{thm:Rado}, we eventually have
\begin{equation}
\ebw = \{\g \in \e\,|\, \mbox{spec}(\g) \prec \bd{v}\}\,.
\end{equation}
This result in its (almost) final form holds not only for $N=3$ but for any fermion number $N$ and any basis set size $d$.
Finally, the most practical form of our ultimate result reads: For any fermion number $N$ and any basis set size $d$, a 1RDM $\g$ is relaxed $\wb$-ensemble $N$-representable, $\g \in \ebw$, if and only if its decreasingly ordered natural occupation numbers fulfill the following constraints
\begin{eqnarray}\label{inr=2}
\lambda_1^\downarrow+ \lambda_2^\downarrow+ \ldots +\lambda_d^\downarrow &= & N\,, \nonumber \\
\lambda_1^\downarrow &\leq & 1\,, \nonumber \\
\lambda_1^\downarrow+ \lambda_2^\downarrow+\ldots +\lambda_N^\downarrow &\leq & N-1 + w_1\,.
\end{eqnarray}
In agreement with the anticipated hierarchical structure of generalized exclusion principle constraints, the condition $\lambda_1^\downarrow \leq  1$ from the setting $r=1$ is part of the conditions \eqref{inr=2} for the case $r=2$, complemented by one additional new constraint.

\subsection{Relaxed $\wb$-ensemble $N$-representability conditions for $r=3$ non-vanishing weights}\label{sec:r=3}
For the case of three weights, there are two lineups of length three (see also Fig.~\ref{fig:excit}),
\begin{eqnarray}
\mbox{(1):}&\quad&(1,2,3)\rightarrow (1,2,4)\rightarrow(1,2,5)\,, \nonumber \\
\mbox{(2):}&\quad&(1,2,3)\rightarrow (1,2,4)\rightarrow(1,3,4)\,.
\end{eqnarray}
The vector of natural occupation numbers of the corresponding $\wb$-minimizers follow as (recall $w_1+w_2+w_3 =1$)
\begin{eqnarray}
\bd{v}^{(1)}&=&(1,1,w_1,w_2,w_3,0,\ldots)\,, \nonumber \\
\bd{v}^{(2)}&=&(1,w_1+w_2,w_1+w_3,w_2+w_3,0,\ldots)\,.
\end{eqnarray}
According to Theorem \ref{thm:Radogener}, it follows that $\g \in \ebw$ if and only if there exists at least one convex combination $\bd{u}\equiv q \bd{v}^{(1)}+(1-q)\bd{v}^{(2)}$ such that $\mbox{spec}(\g)\prec \bd{u}$. To understand how the spectral polytope $\Sigma^\downarrow(\wb)$ looks like we state explicitly this majorization condition:
\begin{eqnarray}\label{PCmajorN3r3}
\lambda_1^\downarrow &\leq& 1\,, \nonumber \\
\lambda_1^\downarrow+\lambda_2^\downarrow &\leq &1+ w_1 + w_2+ q (1-w_1-w_2)\,, \nonumber \\
\lambda_1^\downarrow+\lambda_2^\downarrow +\lambda_3 &\leq &2+w_1\,,  \nonumber \\
\lambda_1^\downarrow+\lambda_2^\downarrow +\lambda_3^\downarrow +\lambda_4^\downarrow&\leq& 3-q (1-w_1-w_2)\,, \nonumber \\
 \sum_{j=1}^{k} \lambda_j^\downarrow &\leq& 3\,,\quad \forall k \geq 5 \,.
\end{eqnarray}
The restrictions on $\sum_{j=1}^{k}\lambda_j^\downarrow$ for $k\geq 5$ are redundant (the normalization of $\g$ implies that they are automatically fulfilled). To determine the spectral polytope $\Sigma^\downarrow(\wb)$ in the minimal hyperplane representation we need to get rid of the parameter $q$. Since $q$ is restricted to $[0,1]$ the upper bound on $\lambda_1^\downarrow+\lambda_2^\downarrow$ can vary between $1 + w_1 + w_2$ and $2$ and the one on $\lambda_1^\downarrow+\lambda_2^\downarrow +\lambda_3^\downarrow +\lambda_4^\downarrow$ between $2+w_1+w_2$ and $3$.
This implies that independent of the value $q$ there is definitely the constraint $\lambda_1^\downarrow+\lambda_2^\downarrow\leq 2$. Moreover, whenever $\lambda_1^\downarrow+\lambda_2^\downarrow$ respects this bound but exceeds the lower one (value $1+w_1+w_2$) it requires us to increase the value $q$ according to
\begin{equation}
\frac{\lambda_1^\downarrow+\lambda_2^\downarrow-1-w_1-w_2}{1-w_1-w_2}\leq q \,.
\end{equation}
This in turn tightens the upper bound (taking the value $3$) on $\lambda_1^\downarrow+\lambda_2^\downarrow+\lambda_3^\downarrow+\lambda_4^\downarrow$ to the value $4+w_1+w_2-\lambda_1^\downarrow-\lambda_2^\downarrow$.
Consequently, the two inequalities in \eqref{PCmajorN3r3} involving the parameter $q$ can equivalently be stated as
\begin{eqnarray}
\lambda_1^\downarrow+\lambda_2^\downarrow &\leq &2\,, \nonumber \\
2 \lambda_1^\downarrow+2\lambda_2^\downarrow+\lambda_3^\downarrow +\lambda_4^\downarrow &\leq& 4 + w_1 + w_2\,.
\end{eqnarray}
Repeating these steps again for \emph{arbitrary} $N$, leads to the final form of our result: For any fermion number $N$ and any basis set size $d$, a 1RDM $\g$ is relaxed $\wb$-ensemble $N$-representable, $\g \in \ebw$, if and only if its decreasingly ordered natural occupation numbers fulfill the following constraints
\begin{eqnarray}\label{inr=3}
\lambda_1^\downarrow+ \lambda_2^\downarrow+ \ldots +\lambda_d^\downarrow &=& N\,, \nonumber \\
\lambda_1^\downarrow &\leq& 1\,, \nonumber \\
\lambda_1^\downarrow+ \lambda_2^\downarrow+\ldots +\lambda_{N}^\downarrow &\leq & N-1 + w_1\,,\nonumber \\
2\sum_{j=1}^{N-1}\lambda_j^\downarrow + \lambda_N^\downarrow+\lambda_{N+1}^\downarrow&\leq & 2(N-1)+w_1+w_2\,.
\end{eqnarray}
In agreement with Table \ref{tab:r}, we have three (non-trivial) facet-defining inequalities for $r=3$.
In particular, the last inequality is the only additional new one compared to $r=2$. This illustrates again the hierarchy of generalized exclusion principle constraints.

\subsection{Relaxed $\wb$-ensemble $N$-representability conditions for $r=4$ non-vanishing weights}\label{sec:r=4}
For the case of four weights, there are four lineups of length four (see also Fig.~\ref{fig:excit}),
\begin{eqnarray}
\mbox{(1):}&\quad&(1,2,3)\rightarrow (1,2,4)\rightarrow(1,2,5)\rightarrow(1,3,4)\,, \nonumber \\
\mbox{(2):}&\quad&(1,2,3)\rightarrow (1,2,4)\rightarrow(1,3,4)\rightarrow(1,2,5)\,, \nonumber \\
\mbox{(3):}&\quad&(1,2,3)\rightarrow (1,2,4)\rightarrow(1,2,5)\rightarrow(1,2,6)\,, \nonumber \\
\mbox{(4):}&\quad&(1,2,3)\rightarrow (1,2,4)\rightarrow(1,3,4)\rightarrow (2,3,4)\,.
\end{eqnarray}
The corresponding natural occupation number vectors $\bd v^{(i)}$ for each lineup are presented in Tab.~\ref{tab:SigmaR4}.
There, we skip various additional entries $0$ in case the one-particle Hilbert space has a dimension $d>6$. The reader shall also
recall that $w_1+w_2+w_3+w_4=1$.

\begin{table}[htb]
\centering
$
\begin{array}{|c|c|c|c|c|}
\hline
&\bd{v}^{(1)}&\bd{v}^{(2)}&\bd{v}^{(3)}&\bd{v}^{(4)} \nonumber \\  \hline
v_1&1&1&1&w_1+w_2+w_3 \nonumber \\ \hline
v_2&w_1+w_2+w_3&w_1+w_2+w_4&1&w_1+w_2+w_4 \nonumber \\ \hline
v_3&w_1+w_4&w_1+w_3&w_1&w_1+w_3+w_4 \nonumber \\ \hline
v_4&w_2+w_4&w_2+w_3&w_2&w_2+w_3+w_4 \nonumber \\ \hline
v_5&w_3&w_4&w_3&0 \nonumber \\ \hline
v_6&0&0&w_4&0 \nonumber \\ \hline
\end{array}
$
\caption{For the setting $(N,d)=(3,6)$ and four weights there are four generating vertices $\bd{v}^{(i)}$ emerging from the $\wb$-minimization of the energy expectation value \eqref{Min1liftN}. The corresponding results are representative for all $(N,d)$ and one would just need to add entries `1' and `0' accordingly.}\label{tab:SigmaR4}
\end{table}

Due to the larger number of generating vertices $\bd v^{(i)}$ compared to the cases $r\leq 3$, turning the vertex representation into a halfspace representation is more tedious and in particular requires additional mathematical tools. We present the mathematical formalism for this in Ref.~\onlinecite{CLLS21} and provide in the following only the inequalities. For the setting $(N,d)=(3,6)$, the minimal hyperplane representation consists of all facet-defining inequalities for $r=3$, complemented by the following two new constraints
\begin{eqnarray}\label{eq:N3r4}
2\lambda_1^\downarrow+ \lambda_2^\downarrow+ \lambda_3^\downarrow+ \lambda_4^\downarrow &\leq & 3+w_1+w_2+w_3\,, \nonumber \\
2\lambda_1^\downarrow+ 2\lambda_2^\downarrow+ \lambda_3^\downarrow+ \lambda_4^\downarrow + \lambda_5^\downarrow &\leq & 4 + w_1 +w_2+w_3 \,.
\end{eqnarray}
Thus, we obtain in total five (non-trivial) constraints on the natural occupation number vector $\bd\lambda^\downarrow$, in agreement with Table \ref{tab:r}. Hence (see also \onlinecite{CLLS21}), for \emph{arbitrary} fermion number $N$ and \emph{arbitrary} basis set size $d$, a 1RDM $\g$ is relaxed $\wb$-ensemble $N$-representable for $r=4$, $\g \in \ebw$, if and only if its decreasingly ordered natural occupation numbers fulfill the following constraints
\begin{widetext}
\begin{eqnarray}\label{inr=4}
\lambda_1^\downarrow+ \lambda_2^\downarrow+ \ldots +\lambda_d^\downarrow &=& N\,, \nonumber \\
\lambda_1^\downarrow &\leq& 1\,, \nonumber \\
\lambda_1^\downarrow+ \lambda_2^\downarrow+\ldots +\lambda_{N}^\downarrow &\leq & N-1 + w_1\,,\nonumber \\
2\sum_{j=1}^{N-1}\lambda_j^\downarrow + \lambda_N^\downarrow+\lambda_{N+1}^\downarrow&\leq & 2(N-1)+w_1+w_2\,, \nonumber \\
2\sum_{i=1}^{N-2}\lambda_i^\downarrow + \lambda_{N-1}^\downarrow + \lambda_N^\downarrow + \lambda_{N+1}^\downarrow &\leq&  2(N-2)+1+w_1+w_2+w_3\,, \nonumber \\
2\sum_{i=1}^{N-1}\lambda_i^\downarrow + \lambda_N^\downarrow + \lambda_{N+1}^\downarrow + \lambda_{N+2}^\downarrow &\leq&  2(N-1) + w_1 +w_2+w_3 \,.
\end{eqnarray}
\end{widetext}

\section{Role of generalized exclusion principle in lattice DFT}\label{sec:latticeDFT}
The relaxed $\wb$-ensemble $N$-representability conditions (generalized exclusion principle constraints) naturally emerged in the context of $\wb$-ensemble RDMFT. Contrarily, in the related $\wb$-ensemble density functional theory (GOK-DFT), the functional's domain is assumed to be not restricted by any non-trivial constraints. This assumption is incorrect, however, at least for discretized models. To explain this,
we consider for the moment spinless fermions on a lattice with $d$ sites. The relaxed universal density functional, denoted by $\Gbw(\bd{n})$ to avoid confusion with the $\wb$-RDMFT functional, depends on the vector $\bd{n}$ of occupation numbers $n_i$ of the lattice site states $\ket{i}$, $i=1,2,\ldots,d$. Similarly to RDMFT, one identifies in DFT a class of Hamiltonians of interest which defines the scope of the corresponding DFT. It is given by
\begin{equation}\label{hamDFT}
H(v)\equiv v+t +W\,,
\end{equation}
i.e., not only the pair interaction $W$ but also the kinetic energy operator $t$ is fixed, while the external potential $v$ remains to be a free  parameter. Consequently, DFT is a special case of RDMFT which is obtained by restricting the family of \emph{all} one-particle Hamiltonians $h$ to the (affine) subclass $h \equiv h(v)\equiv v+t$ with some fixed $t$. To understand the form of the relaxed functional's domain in ground state or more generally $\wb$-ensemble DFT, one needs to address the corresponding relaxed $N$-representability problem: Which occupation number vectors $\bd{n}$ are compatible to some $N$-fermion ensemble state $\G$ with a spectrum majorized by $\wb$? By introducing the function $\mbox{diag}(\cdot)$ which maps 1RDMs $\g$ to their vectors $\bd{n}$ of diagonal entries (with respect to the lattice site basis states $\{\ket{i}\}$) this means to determine the domain $\mbox{dom}(\Gbw)$ of the density functional $\Gbw$ as
\begin{eqnarray}\label{DFTdomain}
\mbox{dom}(\Gbw) &=& \mbox{diag}\Big(N \mbox{Tr}_{N-1}\big[\Ebw\big]\Big) \nonumber \\
& =&  \mbox{diag}\Big(\ebw\Big)\,.
\end{eqnarray}
Based on the insights of our work and the comprehensive introduction into convex analysis, we can answer the question above in a straightforward manner. This involves two steps according to \eqref{DFTdomain}: describing the set $\ebw$ and then understanding the image of the map $\mbox{diag}(\cdot)$. The first problem was comprehensively solved in our work and led to the generalized exclusion principe constraints, illustrated in Sec.~\ref{sec:examples}. For the second step, we just need to recall Theorem \ref{thm:Schur}. Applying it to the 1RDM and the reference basis of lattice site states $\ket{i}$, it states that $\bd{n}$ is compatibly to a 1RDMs $\g$ if and only if $\bd{n}$ is majorized by the eigenvalues of $\g$, i.e.,
\begin{equation}
\bd{n} \equiv (\bra{i}\g \ket{i})_{i=1}^d \prec \mbox{spec}(\g) \equiv \bd{\lambda}\,.
\end{equation}
These two steps can be combined thanks to our generalization of Rado's theorem (Theorem \ref{thm:Radogener}) and by referring to the transitivity property of the majorization $\prec$. The latter means that for any $\bd{u},\bd{v},\bd{w}$ the relations $\bd{u} \prec \bd{v}$ and $\bd{v} \prec \bd{w}$ imply $\bd{u} \prec \bd{w}$. Theorem \ref{thm:Radogener} then yields immediately that $\bd{n}$ lies in the domain of the density functional $\Gbw$ if and only if $\bd{n}$ fulfills the relaxed one-body $\wb$-ensemble $N$-representability constraints.  Hence, our generalized exclusion principle conditions also restrict the domain of GOK-DFT. Yet, it is worth noticing that they are typically much less restrictive than in $\wb$-RDMFT. This is due to the fact that saturating them for some $\bd{n}$ in the context of DFT would in addition require to saturate some of the conditions underlying the majorization $\bd{n} \prec \bd{\lambda}$ (recall Sec.~\ref{sec:major}). This is only possible if the natural orbitals with significant occupation numbers are almost identical to the lattice site states $\ket{i}$. An interesting physical regime for this to happen is the one of strong interactions. For instance in the Hubbard model with on-site interaction $U$ and nearest neighbor hopping of strength $t$, this would correspond to the limit $U/t\gg 1$ describing a Mott insulator.

Let us also briefly comment on functional theories for one-particle Hamiltonians $h$ in \eqref{ham} or external potentials $v$ in \eqref{hamDFT} restricted to the subclass of \emph{spin-independent} operators. In those cases, the corresponding $\wb$-RDMFT or GOK-DFT would be based on a universal functional of the orbital part of the 1RDM and the spin-averaged occupation numbers $n_i= n_{i\uparrow}+n_{i\downarrow}$, respectively. For such a restricted RDMFT one could repeat various steps of our work to determine the corresponding relaxed $\wb$-ensemble $N$-representability constraints on the eigenvalues of the orbital 1RDM $\g_l$, obtained from the full 1RDM $\g$ by tracing out the spin, $\g_l \equiv \mbox{Tr}_s[\g]$. This would then lead to a hierarchical generalization of the Pauli exclusion principle $0 \leq \lambda_i^{(l)} \leq 2$ for orbital natural occupation numbers $\lambda_i^{(l)}$. Working this out shall be a challenge for a different project though. In analogy to the case of spin-dependent functional theories, these conditions on orbital natural occupation numbers would then apply to the corresponding DFT as well (which is nothing else than ordinary GOK-DFT). Consequently, it would restrict the functional's domain in GOK-DFT for electrons on a lattice to a smaller subset as the one described by $0 \leq n_i \leq 2$. Due to the argument above concerning the Schur-Horn theorem it remains a future challenge to understand how significant these additional constraints are for the application of GOK-DFT.

\section{Summary and Conclusions}\label{sec:concl}
In the first part of our work, we provided a more comprehensive introduction into key concepts of convex analysis, such as convex hulls, lower convex envelops, exact convex relaxation, majorization, permutohedra and conjugation. We expect that those concepts will play a pivotal role in the future development of functional theories in general, in particular as far as their foundations are concerned. For instance, a basic understanding of convex conjugation and the geometry of density matrices was already sufficient for disproving and correcting very recently one of the fundamental theorems in RDMFT \cite{S18}.

Also for the main part of our work, those concepts from convex analysis were fundamentally important.
They namely allowed us to develop a solid foundation for $\wb$-RDMFT which was recently proposed in Ref.~\onlinecite{SP21} for targeting excitation energies.
First, we used a generalization of the Ritz variational principle to ensemble density operators with eigenvalues $\bd{w}$ in combination with the constrained search to define a universal functional of the one-particle reduced density matrix. Yet, only by realizing an exact convex relaxation this could be turned into a viable functional theory which shall be called $\wb$-ensemble RDMFT or just $\wb$-RDMFT. This general procedure includes Valone's pioneering work \cite{V80} on ground state RDMFT as the special case $\bd{w}=(1,0,\ldots)$. This remarkable observation provides further evidence for the relevance of convex analysis for the understanding and development of functional theories.
Then, we worked out in a comprehensive manner a methodology for deriving a compact description of the functional's domain. This led to a hierarchy of generalized exclusion principle constraints which we illustrate in great detail.
Most importantly, in analogy to Pauli's exclusion principle and in striking contrast to the generalized Pauli constraints \cite{KL06,AK08,KL09}, these conditions are effectively independent of the particle number $N$ and the dimension $d$ of the underlying one-particle Hilbert space. To be more specific (as it will be explained and proven in a more mathematical context in Ref.~\onlinecite{CLLS21}), calculating these conditions  for some specific (smaller) setting $(N,d)$ would immediately imply the corresponding conditions for arbitrary $N,d$, including the important complete basis set limit, i.e., $d\rightarrow \infty$.

Intriguingly, as explained in Sec.~\ref{sec:latticeDFT} the generalized exclusion principle constraints also affect the application of density functional theory in the context of lattice systems. The same conditions namely apply as well to the vector of occupation numbers. It will be one of the future challenges in $\wb$-ensemble DFT (GOK-DFT) to understand better for which systems those conditions are particularly relevant.

Equipped with the constrained search expression \eqref{Fbwmajor} of the universal functional and the compact description of its domain $\ebw$,
the common process of proposing, testing and improving functional approximations can now be initiated. For instance, one potentially promising route would be to first use a generalization of the Hartree-Fock ansatz to propose the $\wb$-analogue of the Hartree-Fock functional. Then, one may tweak this functional to realize the correct normalization condition between the 2RDM and 1RDM, leading to the $\wb$-analogue of the M\"uller functional \cite{M84,BB02}.

Last but not least, it is worth recalling that the future success of our novel method and RDMFT in general necessitates more efficient numerical minimization schemes.  In that regard, recent developments suggest to upgrade functionals from DFT to RDMFT \cite{M21} and put forward  a Kohn-Sham approach based on hypercomplex numbers to unify DFT and RDMFT \cite{Su21a,Su21b}. From a more conceptual point of view, one may wonder to which degree the limitations of the recent natural orbital-driven RDMFT can actually be overcome. For instance, it has been shown that the distinctive form of the Coulomb interaction leads to discontinuous  high-order derivatives of the off-diagonal entries of 1RDM in spatial representation \cite{C20b}. In turn, these and further related results \cite{Hill83,C20c,Sobolev21,CioSt21} explain the surprisingly slow convergence of common minimization schemes with respect to the  number of natural orbitals and provide ideas for systematic improvements  \cite{C20b,C20c,CioSt21}.

\begin{acknowledgments}
C.S. thanks S.\hspace{0.5mm}Pittalis for inspiring and helpful discussions and M.\hspace{0.5mm}Lehn and G.\hspace{0.5mm}M.\hspace{0.5mm}Ziegler for their support which led to this fruitful collaboration with F.C.~and J.-P.L.. We acknowledge financial support from the German Research Foundation (Grant SCHI 1476/1-1) (J.L., C.S.) and the UK Engineering and Physical Sciences Research Council (Grant EP/P007155/1) (C.S.).
\end{acknowledgments}

\appendix
\section{Proof of Theorem \ref{thm:Radogener}}\label{app:Rado2proof}
Theorem \ref{thm:Radogener} states that the two sets
\begin{equation}
\mathcal{P}\equiv \mbox{conv}\big(\big\{\pi(\bd{v}^{(j)})\,\big|\,j=1,\ldots,R, \pi \in \mathcal{S}^d\big\}\big)
\end{equation}
and
\begin{equation}
\mathcal{P}'\equiv \{\bd{\lambda}\in \RR^d \,|\, \exists \bd{v}=\sum_{j=1}^{R}p_j\bd{v}^{(j)}: \bd{\lambda}\prec \bd{v}\}
\end{equation}
coincide.
We first show $\mathcal{P}'\subset \mathcal{P}$. For this we consider $\bd{\lambda} \in \mathcal{P}'$ which is majorized by some
convex combination $\bd{v}=\sum_{j=1}^{R}p_j\bd{v}^{(j)}$, $\bd{\lambda} \prec \bd{v}$. According to Rado's Theorem \ref{thm:Rado}, $\bd{\lambda}$ is contained in the specific permutohedron $\mathcal{P}_{\bd{v}}$. Since $\mathcal{P}_{\bd{v}}$ is in particular a polytope, we can express $\bd{\lambda}$ as a convex combination of its vertices $\{\pi(\bd{v})\}_{\pi \in \mathcal{S}^d}$. Moreover, each $\pi(\bd{v})= \sum_{j=1}^R p_j \pi(\bd{v}^{(j)})$ is apparently a convex combination of (some) vertices of $\mathcal{P}$. This implies that also $\bd{\lambda}$ can be written as a convex combination of the vertices of $\mathcal{P}$ and therefore $\bd{\lambda}\in \mathcal{P}$. To prove $\mathcal{P}\subset \mathcal{P}'$ let us consider some $\bd{\lambda}\in \mathcal{P}$.
Per assumption we can express $\bd{\lambda}$ as a convex combination of the vertices of $\mathcal{P}$ and find
\begin{equation}
\bd{\lambda} = \sum_{i=1}^{R} \sum_{\pi \in \mathcal{S}^d} q_{i,\pi}\pi(\bd{v}^{(i)}) \equiv \sum_{i=1}^{R} p_i \bd{\lambda}_i\,,
\end{equation}
where for all $i=1,2,\ldots,R$ we introduced $p_i\equiv \sum_{\pi \in \mathcal{S}^d} q_{i,\pi}$ and $\bd{\lambda}_i\equiv \sum_{\pi \in \mathcal{S}^d} q_{i,\pi}\pi(\bd{v}^{(i)})/p_i$. Per construction we have $\bd{\lambda}_i \in \mathcal{P}_{\bd{v}^{(i)}}$ and in particular all $\bd{\lambda}_i$ are correctly normalized. Rado's Theorem \ref{thm:Rado} then implies $\bd{\lambda}_i \prec \bd{v}^{(i)}$ for each $i$ and therefore also
\begin{equation}
  \bd{\lambda}= \sum_{i=1}^{R} p_i \bd{\lambda}_i \prec \sum_{i=1}^{R} p_i \bd{v}^{(i)}\,.
\end{equation}
This last relation means nothing else than $\bd{\lambda} \in \mathcal{P}'$ which finishes the proof.
We also would like to stress that the polytope $\mathcal{P}$ can expressed as
\begin{eqnarray}
\mathcal{P} &=& \bigcup \Big\{\mathcal{P}_{\bd{v}}\,\big|\, \bd{v} =\sum_{j=1}^{R}p_j\bd{v}^{(j)} \Big\}\nonumber \\
&=& \mbox{conv}\Big(\bigcup_{i} \mathcal{P}_{\bd{v}^{(i)}}\Big)\,.
\end{eqnarray}

\bibliography{Refs3}

\end{document}